\documentclass[11pt,letterpaper]{article}

\usepackage[english]{babel}
\usepackage[utf8x]{inputenc}
\usepackage[T1]{fontenc}

\usepackage[margin=1in]{geometry}
\usepackage{amsthm}
\usepackage{amsmath}
\usepackage{amsfonts}
\usepackage{amssymb}
\usepackage{bbm}
\usepackage{graphicx}
\usepackage{tikz}
\usepackage{thm-restate}
\usepackage[colorinlistoftodos]{todonotes}
\usepackage[colorlinks=true, allcolors=blue]{hyperref}
\usepackage[ruled,linesnumbered,algo2e,noend]{algorithm2e}
\usepackage{cleveref}
\usepackage{booktabs}

\SetCommentSty{mycommfont}
\SetKwInput{KwInput}{Input} 
\SetKwInput{KwParameter}{Parameters}                %
\SetKwInput{KwOutput}{Output}              %
\SetKwInput{KwReturn}{Return}              %
\SetKwComment{Comment}{$\triangleright$\ }{}
\SetCommentSty{mycommfont}

\newtheorem{theorem}{Theorem}
\newtheorem{lemma}{Lemma}

\newtheorem{definition}{Definition}
\newtheorem{corollary}{Corollary}
\newtheorem{proposition}{Proposition}
\newtheorem{problem}{Problem}
\newtheorem{invariant}{Invariant}

\newcommand{\defeq}{:=}
\newcommand{\norm}[1]{\left\lVert#1\right\rVert}

\newcommand{\inprod}[2]{\left\langle#1, #2\right\rangle}

\newcommand{\eps}{\epsilon}

\newcommand{\argmax}{\textup{argmax}}
\newcommand{\argmin}{\textup{argmin}} 
\newcommand{\R}{\mathbb{R}}

\newcommand{\N}{\mathbb{N}}
\newcommand{\diag}[1]{\textbf{\textup{diag}}\left(#1\right)}
\newcommand{\half}{\frac 1 2}%

\newcommand{\1}{\mathbf{1}}
\newcommand{\E}{\mathbb{E}}

\newcommand{\ma}{\mathbf{A}}
\newcommand{\ai}{\ma_{i:}}
\newcommand{\aj}{\ma_{:j}}
\newcommand{\id}{\mathbf{I}}

\newcommand{\ogibbs}{\oracle^{\textup{gibbs}}}
\newcommand{\tp}{\tilde{p}}
\newcommand{\poly}{\mathcal{P}}
\newcommand{\mmu}{\mathbf{U}}
\newcommand{\tO}{\widetilde{O}}

\newcommand{\Par}[1]{\left(#1\right)}
\newcommand{\Brack}[1]{\left[#1\right]}
\newcommand{\Brace}[1]{\left\{#1\right\}}
\newcommand{\Abs}[1]{\left|#1\right|}

\newcommand{\oracle}{\mathcal{O}}

\newcommand{\bin}{\textup{bin}}
\newcommand{\bra}[1]{\langle #1|}
\newcommand{\ket}[1]{|#1\rangle}
\newcommand{\orma}{\oracle_{\ma}}

\newcommand{\0}{\mathbf{0}}
\newcommand{\e}{\mathbf{e}}
\newcommand{\orst}{\oracle_{\samptree}}
\newcommand{\mm}{\mathbf{M}}

\newcommand{\Time}{\mathcal{T}}
\newcommand{\Tent}{\Time_{\textup{entry}}}
\newcommand{\normop}[1]{\norm{#1}_{\textup{op}}}

\newcommand{\tZ}{\widetilde{Z}}
\newcommand{\calF}{\mathcal{F}}

\newcommand{\samptree}{\mathsf{SamplerTree}}
\newcommand{\update}{\mathsf{Update}}

\newcommand{\subsum}{\mathsf{SubtreeSum}}
\newcommand{\init}{\mathsf{Init}}

\newcommand{\tZp}{\widetilde{Z}_{\textup{prev}}}
\newcommand{\bu}{\bar{u}}
\newcommand{\bv}{\bar{v}}
\newcommand{\tu}{\tilde{u}}
\newcommand{\tv}{\tilde{v}}
\newcommand{\tw}{\tilde{w}}
\newcommand{\tg}{\tilde{g}}
\newcommand{\hu}{\hat{u}}
\newcommand{\hv}{\hat{v}}
\newcommand{\otest}{\oracle_{\textup{test}}}
\newcommand{\osamp}{\oracle_{\textup{samp}}}
\newcommand{\tsamp}{\mathcal{T}_{\textup{samp}}}
\newcommand{\third}{\frac{1}{3}}
\newcommand{\ninth}{\frac{1}{9}}
\newcommand{\bigset}{\mathcal{B}}
\newcommand{\tq}{\tilde{q}}
\newcommand{\tup}{\mathcal{T}_{\textup{update}}}
\newcommand{\calq}{\mathcal{Q}}

\definecolor{burntorange}{rgb}{0.8, 0.33, 0.0}

\title{Quantum Speedups for Zero-Sum Games \\ via Improved Dynamic Gibbs Sampling}

\author{Adam Bouland ~~~ Yosheb Getachew ~~~ Yujia Jin ~~~  Aaron Sidford ~~~ Kevin 
Tian\footnote{Work partly completed while at Stanford.}\\
	\texttt{\{\href{mailto:abouland@stanford.edu}{abouland},%
	\href{mailto:yoshebg@stanford.edu}{yoshebg},%
		\href{mailto:yujiajin@stanford.edu}{yujiajin},%
		\href{mailto:sidford@stanford.edu}{sidford}%
		\}@stanford.edu}, \texttt{\href{mailto:tiankevin@microsoft.com}{tiankevin}@microsoft.com}}
\date{}

\begin{document}

\maketitle
\thispagestyle{empty}

\begin{abstract}
We give a quantum algorithm for computing an $\epsilon$-approximate Nash equilibrium of a zero-sum game in a $m \times n$ payoff matrix with bounded entries. Given a standard quantum oracle for accessing the payoff matrix our algorithm runs in time $\widetilde{O}(\sqrt{m + n}\cdot \epsilon^{-2.5} + \epsilon^{-3})$ and outputs a classical representation of the $\epsilon$-approximate Nash equilibrium. This improves upon the best prior quantum runtime of $\widetilde{O}(\sqrt{m + n} \cdot \epsilon^{-3})$ obtained by \cite{vanApeldoornG19} and the classic $\widetilde{O}((m + n) \cdot \epsilon^{-2})$ runtime due to \cite{GrigoriadisK95} whenever $\epsilon = \Omega((m +n)^{-1})$. We obtain this result by designing new quantum data structures for efficiently sampling from a slowly-changing Gibbs distribution. 
\end{abstract}

\newpage
\clearpage
\pagenumbering{arabic}

\section{Introduction}
\label{sec:intro}

There is now a broad family of quantum algorithms for machine learning and fast numerical linear algebra \cite{biamonte2017quantum}, built on many quantum algorithmic primitives, e.g.\ \cite{BrassardHMT02,harrow2009quantum,GilyenSLW19}.
More specifically, for a wide range of problems it has been shown how quantum algorithms can (in certain parameter regimes) yield faster runtimes.\footnote{Note that quantifying the end-to-end speedups obtained by these methods can be subtle due to I/O overheads, different access models \cite{aaronson2015read}, and classical de-quantization algorithms \cite{tang2019quantum,chia2020sampling,gharibian2022dequantizing}.}
These quantum algorithms obtain runtimes which improve upon the dimension dependence of classical algorithms, but often at the cost of a worse dependence on the error tolerance and/or 
implicit access to the solution (e.g.\ query or sampling access for solution entries). Consequently, this paper is motivated by the following question.

\begin{quote}
\emph{To what degree is there an inherent accuracy versus dimension-dependence tradeoff for quantum optimization algorithms? What algorithmic techniques improve this tradeoff? }
\end{quote} 

In this paper we consider this question for the fundamental optimization problem of computing $\epsilon$-approximate Nash equilibrium in zero-sum games. Our main result is an improved dependence on $\eps$ for quantum algorithms solving zero-sum games, which is very close to that of its classical counterpart. Further, we show that for our algorithms, obtaining a classical representation of the solution is obtainable at no additional asymptotic cost. Our work builds upon \cite{vanApeldoornG19,li2019sublinear}, which already took a large and important step towards answering the above question by designing quantum data structures for efficiently implementing algorithms for solving zero-sum games.

Interestingly, to obtain our result we provide improved quantum algorithms for solving a dynamic data structure problem of sampling from a slowly-changing Gibbs distribution. Such dynamic sampling problems arise as a natural component of stochastic gradient methods for solving zero-sum games. We obtain our speedups by improving a Gibbs sampling subroutine developed in \cite{vanApeldoornG19}. We design a new dynamic quantum data structure which performs the necessary Gibbs sampling in time $\tO(\eps^{-\half})$, which is faster than the corresponding $\tO(\eps^{-1})$ runtime achieved by \cite{vanApeldoornG19}. Beyond the intrinsic utility of solving this problem, we hope our improved Gibbs sampler showcases potential algorithmic insights that can be gleaned by seeking improved error dependencies for quantum optimization algorithms. Moreover, we hope this work encourages the study and design of quantum data structures for efficient optimization.

\subsection{Zero-sum games}
For matrix $\ma \in \R^{m \times n}$ its associated \emph{zero-sum game} is the pair of equivalent optimization problems
\[
\min_{u \in \Delta^m} \max_{v\in\Delta^n} u^\top \ma v = \max_{v\in\Delta^n} \min_{u \in \Delta^m}  u^\top \ma v,
\text{ where }
\textstyle{
\Delta^k \defeq \{x \in \R^k_{\geq 0} : \sum_{i \in [k]} x_i = 1\}
}.
\]
In such a game, we refer to $\ma$ as the \emph{payoff matrix} and view the $m$ and $n$-dimensional simplices, i.e.\ $\Delta^m$ and $\Delta^n$, as the space of distributions over $[m]$ and $[n]$ respectively. From this perspective $u^\top \ma v$, known as \emph{payoff} or \emph{utility} of $(u,v)$, is the expected value of $\ma_{ij}$ when sampling $i \in [m]$ and $j \in [n]$ independently from the distributions corresponding to $u$ and $v$. Thus, a zero-sum game models a two-player game where a minimization player seeks to minimize the payoff while, simultaneously, a maximization player seeks to maximize it.

In this paper, we consider the canonical problem of computing an approximate Nash equilibrium of a zero-sum game. Given the payoff matrix $\ma \in \R^{m \times n}$ we call a pair $(u, v) \in \Delta^m \times \Delta^n$ an \emph{$\epsilon$-approximate Nash equilibrium (NE)} for $\eps \in \R_{> 0}$  if
\[
\Par{\max_{v' \in \Delta^n} u^\top \ma v' } - \Par{\min_{u' \in \Delta^m} (u')^\top \ma v} \le \eps.
\]
We assume that the payoff matrix $\ma$ and the error-tolerance are given as input to an algorithm, and that, for simplicity, $\norm{\ma}_{\max} \le 1$, i.e.\ the largest entry of $\ma$ has magnitude at most $1$ (this is without loss of generality by rescaling  $\ma \gets \norm{\ma}_{\max}^{-1} \ma$ and $\eps \gets \norm{\ma}_{\max}^{-1} \eps$). The main goal of this paper is to design improved zero-sum game solvers, i.e.\ algorithms that compute $\epsilon$-approximate NEs.

Zero-sum games are foundational to theoretical computer science, optimization, and economics. The problem of approximately solving zero-sum games is a natural formulation of approximate linear programming (LP) and correspondingly, this problem is a prominent testbed for new optimization techniques. Over the past decades there have been numerous advances in the computational complexity of solving zero-sum games under various assumptions on problem parameter (see Section \ref{sec:related-work} for a survey). Recent advancements in interior point methods (IPMs) for linear programming, e.g.\ \cite{BrandLLSSSW21} and references therein (discussed in more detail in Section~\ref{sec:related-work}), solve zero sum-games  in time $\tO(mn + \min(m, n)^{2.5})$.\footnote{We use the $\tO$ notation to hide polylogarithmic dependences on problem parameters when convenient for exposition; see Section~\ref{sec:prelims} for a more detailed statement of hidden parameters. In informal theorem statements, we use ``with high probability'' to indicate a polylogarithmic dependence on the failure probability.} Further the linear programming algorithm of \cite{Brand20}, shows that zero-sum games can be solved deterministically in $\tO((m + n)^\omega)$ time where $\omega < 2.373$ is the current matrix multiplication constant \cite{AlmanW21}, or $\tO((m + n)^3)$ without fast matrix multiplication. In this paper, we primarily focus on sublinear-time algorithms for approximating NEs. 

A well-known algorithm by \cite{GrigoriadisK95} achieves a runtime of $\tO((m + n) \cdot \eps^{-2})$, which is the state-of-the-art sublinear runtime amongst classical algorithms, without further problem assumptions. Recently it has been shown that quantum algorithms can yield strikingly runtime improvements for solving zero-sum games and their generalizations \cite{li2019sublinear, vanApeldoornG19,li2021sublinear}.
In particular, in 2019 Li, Chakrabati and Wu \cite{li2019sublinear} gave a quantum algorithm for zero sum games in time $\tO(\sqrt{m+n} \cdot \epsilon^{-4})$, and simultaneously van Apeldoorn and Gilyen \cite{vanApeldoornG19} gave an algorithm running in time $\tO(\sqrt{m+n} \cdot \epsilon^{-3})$.
These algorithms yield a quadratic improvement in the dimension dependence of the best classical algorithm, at the cost of a higher error dependence.

The algorithms of \cite{li2019sublinear, vanApeldoornG19,li2021sublinear} operate using a standard quantum oracle for $\ma$ (formally stated in \Cref{sec:prelims}), in which one can query the entries of $\ma$ in superposition.
We focus on the algorithm of \cite{vanApeldoornG19} for the rest of this paper, as we focus on improving error dependence.
The \cite{vanApeldoornG19} algorithm generalizes the classical algorithm of Grigoriadis and Khachiyan \cite{GrigoriadisK95}, and obtains a runtime improvement by speeding up a key dynamic Gibbs sampling subroutine required by the \cite{GrigoriadisK95} method. As we discuss in greater detail in Section~\ref{sec:overview}, van Apeldoorn and Gilyen give a quantum data structure to efficiently perform this sampling in time quadratically faster in the dimension, which lies at the core of their algorithmic speedup.

\paragraph{Our result.} We give a new quantum algorithm for solving zero-sum games which improves upon the runtime of the prior state-of-the-art quantum algorithm, due to \cite{vanApeldoornG19}.

\begin{theorem}[informal, see Theorem~\ref{thm:quantumgames}]\label{thm:main_informal}
	Let $\ma \in \R^{m \times n}$ with $\norm{\ma}_{\max} \le 1$, and $\eps \in (0, 1)$. Given a quantum oracle for $\ma$ (defined in \Cref{sec:prelims}), there is an $\tO(\sqrt{m + n} \cdot \epsilon^{-2.5} + \eps^{-3})$ time algorithm which yields a classical output $(u, v) \in \Delta^m \times \Delta^n$ that is an $\eps$-approximate NE with high probability.
\end{theorem}

Our new algorithm simultaneously improves the best known quantum \cite{vanApeldoornG19} and classical \cite{GrigoriadisK95} algorithms in the parameter regime where IPMs do not dominate sublinear algorithms. In particular, it is faster than the classical  $\tO((m + n) \cdot \eps^{-2})$ runtime of \cite{GrigoriadisK95} whenever $\epsilon^{-1} = \tO(m + n)$, which includes the regime where \cite{GrigoriadisK95} offers advantages over the $\tO((m + n)^\omega)$ runtime of the \cite{Brand20} IPM, as $\omega < 3$. This is in contrast to the prior quantum rate of \cite{vanApeldoornG19}, which does not achieve an improvement upon \cite{GrigoriadisK95} in the full parameter range where sublinear algorithms are currently preferable to IPMs. For example, when $m \approx n$ and (up to logarithmic factors) $\eps \in [n^{-c}, n^{-\half}]$ where $c = \half(\omega - 1)$, the rate of \cite{GrigoriadisK95} is favorable to that of  \cite{vanApeldoornG19} and state-of-the-art IPMs \cite{Brand20, cohen2021solving}.\footnote{There is evidence that $\omega = 2$ cannot be achieved with current techniques, e.g.\ \cite{Alman21}.}

\subsection{Dynamic Gibbs sampling}

We obtain the improved error dependence in our zero-sum game solver by producing a new, faster quantum data structure to perform the Gibbs sampling as used in the algorithm of \cite{vanApeldoornG19}, which may be of independent interest. Gibbs sampling is a fundamental algorithmic primitive --- the basic task is, given vector $v\in \mathbb{R}^n$, sample from the probability distribution proportional to $\exp(v)$. Gibbs sampling is used as a subroutine in many quantum and classical optimization algorithms, e.g.\ \cite{brandao2017quantum} and follow-up works. In general, quantum algorithms can perform this task more efficiently using amplitude estimation, which can boost the acceptance probability of rejection sampling schemes. This strategy was implemented in \cite{vanApeldoornG19}, which approximate the maximum entry $v_{\max}$ of $v$ using quantum maximum finding~\cite{DurrH96}, uniformly sample $i \in [n]$, and accept the sample with probability $\exp(v_i-v_{\max})\le 1$ using quantum rejection sampling. We give a more detailed overview of the \cite{vanApeldoornG19} Gibbs sampler and its complexity analysis in Section~\ref{ssec:hintsample_overview}.

We give a data structure which quadratically improves the error dependence of the \cite{vanApeldoornG19} Gibbs sampling subroutine runtime, from $\tO(\sqrt{m + n} \cdot \epsilon^{-1})$ per sample to an amortized $\tO (\sqrt{m+n} \cdot \epsilon^{-\half})$ per sample. A key fact which enables this improvement is that the Gibbs distributions one samples from in the zero-sum game solver of \cite{GrigoriadisK95} change slowly over time: the base vector $v$ receives bounded sparse updates in each iteration.
By storing partial information about the Gibbs distribution, namely an efficiently-computable overestimate to its entries which remains valid across many consecutive iterations, we obtain an improved dynamic Gibbs sampler, which we also provide a detailed overview of in Section~\ref{ssec:hintsample_overview}.

We now define our notion of an approximate Gibbs sampler, and then state the dynamic sampling problem we consider, which arises naturally in zero-sum game algorithms with sublinear runtimes.

\begin{definition}[Approximate Gibbs oracle]\label{def:ogibbs}
	For $v \in \R^n$, its associated \emph{Gibbs distribution} is $p_v \in \Delta^n$ such that for all $i \in [n]$, $[p_v]_i \propto \exp(v_i)$. We say $\ogibbs_v$ is a $\delta$-approximate Gibbs oracle if it samples from $\tp \in \Delta^n$ with $\norm{\tp - p_v}_1 \le \delta$.
\end{definition}

\begin{problem}[Sampling maintenance]\label{prob:sample}
	Let $\eta > 0$, $\delta \in (0, 1)$, and suppose we have a quantum oracle for $\ma \in \R^{m\times n}$. Consider a sequence of $T$ $\update$ operations to a dynamic vector $x \in \R^m_{\ge 0}$, each of the form $x_i \gets x_i + \eta$ for some $i \in [m]$. In the \textup{sampling maintenance problem}, in amortized $\tup$ time per $\update$ we must maintain a $\delta$-approximate Gibbs oracle, $\osamp$, for $\ma^\top x$ which is queryable in worst-case time $\tsamp$.
\end{problem}

\paragraph{Our result.} We provide a quantum algorithm for solving Problem~\ref{prob:sample}, which improves upon the runtime implied by the corresponding component in the algorithm of \cite{vanApeldoornG19}.

\begin{theorem}[informal, see Theorem~\ref{prop:mainhintmaintain}]\label{prop:prob1_informal}
There is a quantum algorithm which solves \Cref{prob:sample} with high probability
with
$\max(\tsamp, \tup) = \tO\Par{\sqrt{n} \cdot T\eta^{1.5}}$ and an initialization cost of $\tO\Par{\eta^3T^3}$.
\end{theorem}

 \Cref{prop:prob1_informal} improves upon the solution to the sampling maintenance Problem~\ref{prob:sample} implied by \cite{vanApeldoornG19} by a $\eta^{-\half}$ factor; in the setting of the \cite{GrigoriadisK95} solver, where $T = \tO(\eps^{-2})$ and $\eta = \Theta(\eps)$, this is an $\eps^{-\half}$-factor improvement. At a high level, our improvement is obtained by storing a hint consisting of a vector which overestimates the true Gibbs distribution, and an approximate normalization factor, which are infrequently updated. Our maintained hint satisfies the desirable properties that: $(i)$ it remains valid for a batch of consecutive iterations, and $(ii)$ the degree of overestimation is bounded. The former property ensures a fast amortized update time, and the latter ensures a fast sample time by lower bounding the acceptance probability of our quantum rejection sampler. Our high-level strategy for maintaining improved hints is to repeatedly call our sampling access to accurately estimate large entries of the Gibbs distribution, and to exploit stability of the distribution under the setting of Problem~\ref{prob:sample}.
We discuss our dynamic Gibbs sampler in more detail and compare it with previous methods for solving~\Cref{prob:sample} in Section \ref{ssec:hintsample_overview}.

The initialization cost of Theorem~\ref{prop:prob1_informal} is due to the current state-of-the-art in numerically stable implementations of the quantum singular value transformation (SVT) framework of \cite{GilyenSLW19}. This cost is also the cause of the additive $\tO(\eps^{-3})$ term in Theorem~\ref{thm:main_informal}. We discuss this cost in Appendix~\ref{app:stablepoly}; improvements to numerically stable implementations of \cite{GilyenSLW19} would be reflected in the runtimes of Theorems~\ref{thm:main_informal} and~\ref{prop:prob1_informal}.

\subsection{Related work}
\label{sec:related-work}

\begin{table}
	\centering
	\renewcommand{\arraystretch}{1.25}
	\everymath{\displaystyle}
	\caption{\textbf{Algorithms for computing $\eps$-approximate Nash equilibria of
		zero-sum games.} Hides polylogarithmic factors and assumes $\ma \in \R^{m \times n}$ with $\norm{\ma}_{\max} \le 1$. }\label{table:zero-sum-runtimes}
	\begin{tabular}{c c c}	
		\toprule
		{ Method } & {Query model} & {Total runtime} \\
		
		\midrule
		interior point method~\cite{cohen2021solving} & classical & $\max(m,n)^\omega$ \\
		interior point method~\cite{BrandLLSSSW21} & classical & $mn + \min(m, n)^{2.5}$ \\
		\midrule
		extragradient~\cite{Nemirovski04,Nesterov07} & classical &  $ 
		mn \cdot \eps^{-1}$ 
		\\ 
		stochastic mirror descent (SMD)~\cite{GrigoriadisK95} & 
		classical &  $ 
		(m + n) \cdot \eps^{-2}$ 
		\\
		variance-reduced SMD~\cite{CarmonJST19} & classical & $mn + \sqrt{mn(m + n)} \cdot \eps^{-1}$
		\\
		\midrule
		\cite{vanApeldoornG19} & quantum & 
		$\sqrt{m + n} \cdot \eps^{-3}$  \\
		Theorem~\ref{thm:main_informal} (\textbf{our work})  & quantum & $\sqrt{m + n} \cdot \eps^{-2.5} + \eps^{-3}$  \\
		\bottomrule
	\end{tabular}
\end{table}

\begin{table}
	\centering
	\renewcommand{\arraystretch}{1.25}
	\everymath{\displaystyle}
	\caption{\textbf{Solutions to Problem~\ref{prob:sample}, $T = \eps^{-2}$, $\eta = \eps$.} Hides polylogarithmic factors.}\label{table:dynamic-gibbs-runtimes}
	\begin{tabular}{c c c c}	
		\toprule
		{ Method } & {Query model} & {$\tsamp$} & {$\tup$} \\
		\midrule
		explicit updates~\cite{GrigoriadisK95} & classical & $1$ & $m + n$ \\
		max-based rejection sampling~\cite{vanApeldoornG19} & quantum &  $ 
		\sqrt{m + n} \cdot \eps^{-1}$ & $ 
		\sqrt{m + n} \cdot \eps^{-1}$ 
		\\ 
		\Cref{prop:prob1_informal} (\textbf{our work}) & quantum & $ 
		\sqrt{m + n} \cdot \eps^{-\half}$ & $ 
		\sqrt{m + n} \cdot \eps^{-\half}$  \\
		\bottomrule
	\end{tabular}
\end{table}

\paragraph{Quantum optimization and machine learning.} There are  a wide array of quantum algorithms for optimization and machine learning which make use of fundamental algorithmic primitives such as amplitude amplification \cite{BrassardHMT02}, the HHL algorithm \cite{harrow2009quantum}, and the quantum singular value transformation \cite{GilyenSLW19}. For example, a number of works gave HHL-based algorithms for a variety of machine learning tasks such as PCA \cite{lloyd2014quantum}, SVMs \cite{rebentrost2014quantum}, and recommendation systems \cite{kerenidis2016quantum}. For more details see the survey article of \cite{biamonte2017quantum}.

Most relevant to our current work are quantum algorithms for optimization problems.
For example, Brandao and Svore \cite{brandao2017quantum} gave a quantum algorithm for SDP solving based on the Arora-Kale algorithm \cite{arora2007combinatorial}, which was later improved by \cite{van2020quantum}.
There have also been quantum IPM-based methods for LPs and SDPs \cite{kerenidis2020quantum}.
Additionally a series of works have considered quantum algorithms for general convex optimization~\cite{chakrabarti2020quantum,van2020convex}, which make use of Jordan's algorithm for fast gradient estimation \cite{jordan2005fast,gilyen2019optimizing}.

In the area of zero-sum games, in addition to the works previously mentioned \cite{vanApeldoornG19,li2019sublinear} on $\ell_1$-$\ell_1$ games (where both players are $\ell_1$-constrained), there have been several works considering different variants of zero-sum games.
For example Li, Chakrabati and Wu \cite{li2019sublinear} gave quantum algorithms for $\ell_2$-$\ell_1$ games with quadratic improvement on the dimension.
Later Li, Wang, Chakrabati and Wu \cite{li2021sublinear} extended this algorithm to more general $\ell_q$-$\ell_1$ games with $q\in(1,2]$.

\paragraph{Zero-sum games.} Zero-sum games are a canonical modeling tool in optimization, economics and machine learning~\cite{VonNeumann28}. The classic extragradient (mirror prox) method~\cite{Nemirovski04,Nesterov07} computes an $\eps$-approximate NE in $\tO(mn \cdot \epsilon^{-1})$ time; as discussed previously, the stochastic mirror descent method of \cite{GrigoriadisK95} obtains the same accuracy in  time $\tO((m + n) \cdot \eps^{-2})$. An intermediate runtime was recently obtained by \cite{CarmonJST19} using variance reduction, described in Table~\ref{table:zero-sum-runtimes}.

Improved runtimes are available under more fine-grained characterizations of the matrix $\ma$, such as sparsity (e.g.\ number of nonzero entries per row or column) or numerical sparsity (e.g.\ rows and columns with bounded $\ell_1$-to-$\ell_2$ norm ratios) \cite{CarmonJST20}. Notably, the \cite{GrigoriadisK95} algorithm also offers runtime improvements under a sparsity assumption, as does the algorithm of \cite{vanApeldoornG19} in certain sparsity-to-accuracy ratio regimes. In this paper, we focus on NE algorithms in the general setting (without further sparsity or numerical sparsity assumptions).

In parallel, a long line of research improving IPMs for solving linear programming~\cite{karmarkar1984new,renegar1988polynomial,lee2014path,lee2019solving,van2020solving,JiangSWZ21} has led to a number of different zero-sum game solvers with polylogarithmic runtime dependencies on the problem accuracy $\epsilon$. The current state-of-the-art variants of IPMs are~\cite{cohen2021solving} and~\cite{BrandLLSSSW21}, which achieve runtimes of $\widetilde{O}(\max(m,n)^\omega)$ and $\widetilde{O}(mn + \min(m, n)^{2.5})$ respectively. We refer readers to~\Cref{table:zero-sum-runtimes} for detailed comparisons. Finally, for strongly polynomial runtimes (i.e.\ with no dependence on $\eps$), which are outside the scope of this paper, we refer readers to \cite{DadushNV20} and references therein.

\subsection{Future work} Theorem~\ref{thm:main_informal}'s $\eps$ dependence is within an $\eps^{-\half}$ factor of matching classical counterparts. To the best of our knowledge, removing this $\eps^{-\half}$ overhead would represent the first quantum algorithm for a natural optimization problem which improves upon classical counterparts across all parameters. 

Both our work and \cite{vanApeldoornG19} solve Problem~\ref{prob:sample} by leveraging a powerful polynomial approximation-based technique developed in \cite{GilyenSLW19}, known as the quantum singular value transform (QSVT). In both cases, QSVT is used with a polynomial of degree $\tO(\eps^{-1})$. We note that in closely-related classical settings (discussed in \cite{SachdevaV14}), Chebyshev polynomial-based approximations yield a quadratically smaller degree. However, a boundedness requirement (due to the spectra of quantum gates) prevents straightforwardly applying these constructions within QSVT. Sidestepping this barrier is a natural avenue towards improving our work, which we leave as an open problem.

More generally, establishing optimal oracle query complexities of dynamic Gibbs sampling (e.g.\ Problem~\ref{prob:sample}) and solving zero-sum games are key problems left open by our work. These questions are potentially more approachable than establishing tight time complexity characterizations. For example, could $\max(\tsamp, \tup)$ be improved to $\tO(\sqrt n)$ in the context of Theorem~\ref{thm:main_informal}, or can we rule out such an improvement in the query model?

\subsection{Organization} In Section~\ref{sec:prelims} we state the notation used throughout the paper, as well as the (classical and quantum) computational models we assume. In Section~\ref{sec:overview}, we give a brief technical overview of the core components of our algorithm used to prove Theorem~\ref{thm:main_informal}: the stochastic gradient method our method is built on, and an efficient quantum implementation of a key subroutine using a new dynamic Gibbs sampler. Finally in Section~\ref{sec:oracle} we give our new quantum sampler, and prove Theorem~\ref{prop:prob1_informal}.

We aim to give a self-contained, but simplified, description of our algorithm in Section~\ref{sec:overview} to improve the readability of the paper for readers with an optimization background unfamiliar with quantum computing, and vice versa. In particular, we abstract away the core optimization machinery (stochastic mirror descent) and quantum machinery (quantum SVT) developed in prior work into the statements of Propositions~\ref{prop:main_sgd} and~\ref{prop:mainhintsample}, and focus on how we use these statements black-box to build a faster algorithm. The proofs of these statements can be found in Appendices~\ref{sec:sgd} and~\ref{sec:quantum}. %

\section{Preliminaries}
\label{sec:prelims}

\paragraph{General notation.} $\tO$ hides logarithmic factors in problem dimensions (denoted $m$ and $n$), target accuracies (denoted $\epsilon$), and failure probabilities (denoted $\alpha$). When discussing runtimes for Problem~\ref{prob:sample}, we additionally use $\tO$ to hide logarithmic factors in the parameters $\eta, T$. For all $i \in [n]$ we let $e_i \in \R^n$ denote the $i^{\text{th}}$ standard basis vector for $i \in [n]$ when $n$ is clear. $\norm{\cdot}_p$ denotes the $\ell_p$ norm of a vector. For $\ma \in \R^{m \times n}$, its $i^{\text{th}}$ row and $j^{\text{th}}$ column are respectively $\ai, \aj$. For $v \in \R^n$, $\diag{v}$ is the diagonal $n \times n$ matrix with $v$ as the diagonal. Conjugate transposes of $\ma$ are denoted $\ma^*$; when the matrix is real we use $\ma^\top$. The all-ones and all-zeros vectors of dimension $n$ are $\1_n$ and $\0_n$. Finally, throughout $a \defeq \lceil\log_2 m\rceil$ and $b \defeq \lceil \log_2 n \rceil$, so $[m] \subseteq [2^a]$ and $[n] \subseteq [2^b]$.

\paragraph{Computation models.} We assume entries of $\ma$ are $w$-bit reals for $w = O(\log (mn))$, and work in the word RAM model where $w$-bit arithmetic operations take $O(1)$ time; for simplicity, we assume mathematical operations such as trigonometric functions and radicals can also be implemented exactly for $w$-bit words in $O(1)$ time. Throughout, ``quantum states'' mean unit vectors, and ``quantum gates'' or ``oracles'' $\oracle$ mean unitary matrices. We follow standard notation and identify a standard basis vector $e_i$ for $i \in [n]$ with $\ket{i}$, an $a$-qubit state, in which $i$ is represented in binary (i.e.\ more formally, $\ket{i}=\ket{\bin(i)}$, and bin is omitted for brevity). We consider the standard model of quantum access to oracles, in which the oracle $\oracle$, which is defined by its operation on $|s\rangle$ for all $\{0,1\}^*$-valued $s$ (where length is clear from context), can be queried in superposition. If $\oracle$ is queried on $\ket{v} \defeq \sum_{s} \alpha_s \ket{s}$, the result is $\oracle \ket{v} = \sum_{s} \alpha_i (\oracle \ket{s})$. We use $\ket{g}$, $\ket{g'}$, etc.\ (when clear from context) to denote arbitrary sub-unit vectors, which represent garbage states (unused in computations). The tensor product of states $\ket{u}$ and $\ket{v}$ on $a$ and $b$ qubits is denoted $\ket{u}\ket{v}$, an $(a+b)$-qubit state. The runtime of a quantum circuit is its maximum depth (in arithmetic gates on $w$-bit words). 

\paragraph{Access model.} Throughout the paper, we assume a standard quantum oracle for accessing $\ma$ (recall $\norm{\ma}_{\max} \le 1$). In particular, by a quantum oracle for $\ma$ we mean an oracle $\orma$ which, when queried with $\ket{i}\ket{j}\ket{s}$ for $i \in [m], j \in [n], s \in \{0, 1\}^w$, reversibly writes $\ma_{ij}$ (in binary) to the third register in $O(1)$ time, i.e.\ 
$\orma \ket{i}\ket{j}\ket{s} = \ket{i}\ket{j}\ket{s \oplus \ma_{ij}}$ where $\oplus$ is bitwise mod-$2$ addition.

Given a quantum oracle for $\ma$, with two queries, by standard constructions one can construct an oracle which places the value in the amplitude of the state rather than the register itself.
More formally, one can construct\footnote{This follows e.g.\ by calling the oracle to obtain the value of $\ma_{ij}$ in binary (interpreted as a signed number between $0$ and $1$), adding an ancilla qubit, performing arithmetric to compute the rotation angle needed on that ancilla, applying a tower of controlled rotation gates to an ancilla qubit using that rotation angle express in binary, then calling the standard oracle a second time to uncompute the  binary value of $\ma_{ij}$. See e.g.\ \cite{GroverR02} for details.} an $\orma'$, which operates as:
\[\orma' \ket{0}\ket{i} \ket{j} = \sqrt{\ma_{ij}} \ket{0} \ket{i}\ket{j} + \sqrt{1 - |\ma_{ij}|} \ket{1} \ket{g}, \text{ for } (i, j) \in [m] \times [n].\]
It is standard in the literature to (using ancilla qubits to store the output register where $\ma_{ij}$ is written) construct such an $\orma'$ from $\orma$ under our classical model of computation, see e.g.\ \cite{GroverR02}. For simplicity, we omit discussion of ancilla qubits in the remainder of the paper and assume direct access to $\orma'$. We also note that there is ambiguity in the implementation of $\orma'$ in that the square root is not unique, and that we have control over the signing used in this implementation. We will use this flexibility crucially later in the paper, specifically Corollary~\ref{cor:blockencodediag}.

\section{Overview of approach}
\label{sec:overview}

In this section, we give an overview of the approach we take to prove our main results: an improved quantum runtime for solving zero-sum games (\Cref{thm:quantumgames}) and an improved quantum data structures for dynamic Gibbs sampling (\Cref{prop:mainhintmaintain}). We organize this section as follows.

In Section~\ref{ssec:sgd_overview}, we state Algorithm~\ref{alg:main}, the optimization method framework we use to solve zero-sum games. This framework is a generalization of the classical algorithm of \cite{GrigoriadisK95}. We state its guarantees in Proposition~\ref{prop:main_sgd} and defer the proof to Appendix~\ref{sec:sgd}. Algorithm~\ref{alg:main} assumes access to an approximate Gibbs oracle (Definition~\ref{def:ogibbs}) for sampling from dynamic distributions as stated in Problem~\ref{prob:sample}. The bulk of our work is devoted to obtaining an efficient quantum implementation of such an oracle  (\Cref{prop:mainhintmaintain}) and using this result we prove Theorem~\ref{thm:quantumgames} at the end of Section~\ref{ssec:sgd_overview}.

In Section~\ref{ssec:hintsample_overview}, we overview the main technical innovation of this paper, an improved solution to Problem~\ref{prob:sample}. Whereas prior work by \cite{vanApeldoornG19} solves Problem~\ref{prob:sample} at an amortized $\approx \sqrt{m + n} \cdot \eps^{-1}$ cost per iteration, we show how to solve the problem at an amortized $\approx \sqrt{m + n} \cdot \eps^{-\half}$ cost. We remark that the only quantum components of our algorithm (quantum SVT and amplitude amplification) are abstracted away by Proposition~\ref{prop:mainhintsample}, which is proven in Appendix~\ref{sec:quantum}.

\subsection{Solving matrix games with a Gibbs sampling oracle}\label{ssec:sgd_overview}

Our proof of Theorem~\ref{thm:quantumgames} uses an efficient implementation of the algorithmic framework stated in \Cref{alg:main}, based on stochastic mirror descent. In specifying \Cref{alg:main}, we recall our earlier Definition~\ref{def:ogibbs}, which captures the approximate sampling access we require for Algorithm~\ref{alg:main}'s execution.

\begin{algorithm2e}[H]
	\caption{\textsf{MatrixGameSolver}($\delta, \eta, T$)}
	\label{alg:main}
	\DontPrintSemicolon
	{\bfseries Input:} $\ma\in\R^{m\times n}$, desired accuracy $\epsilon \in (0, 1)$, $\delta$-approximate Gibbs oracles for the (dynamic) vectors $-\ma^\top x_t$ and $\ma y_t$ \;
	{\bfseries Parameters:} Gibbs sampler parameter $\delta \in (0, 1)$, step size $\eta>0$, iteration count $T$ \;
	Initialize $\hat{u} \gets \0_m$, $\hat{v} \gets \0_n$, $x_0\gets \0_m$, and $y_0\gets\0_n$\;
	\For{$t=0$ {\bfseries{\textup{to}}} $T-1$}{
		Independently sample $j_t, j'_t \in [n]$ using $\ogibbs_{-\ma^\top x_t}$ and $i_t, i'_t \in [m]$ using $\ogibbs_{\ma y_t}$ \label{line:gibbs_sample}\;
		Update $y_{t+1}\gets y_t+\eta e_{j_t}$ and $x_{t+1}\gets x_t+\eta e_{i_t}$ \tcp*{Update iterates.}
		Update $\hu \gets \hu + \frac 1 T e_{i'_t}$ and $\hv \gets \hv + \frac 1 T e_{j'_t}$  \tcp*{Update output.}
	}
	\Return $(\hat{u}, \hat{v})$\;
\end{algorithm2e}

The main skeleton of Algorithm~\ref{alg:main} (Lines 5-6) using exact oracles is identical to the method of \cite{GrigoriadisK95}. However, our framework builds upon  \cite{GrigoriadisK95} in the following three ways.

\begin{enumerate}
	\item We tolerate total variation error in the sampling procedure via $\delta$-approximate Gibbs oracles.
\item We provide a high-probability guarantee on the duality gap using martingale arguments.
\item We subsample the output to obtain a sparse solution yielding a comparable duality gap.
\end{enumerate}

We remark that several of these improvements have appeared previously, either explicitly or implicitly, in the stochastic gradient method literature. For example, an approximation-tolerant stochastic gradient method was given in \cite{CarmonJST20}, and our proofs of the high-probability guarantees are based on arguments in \cite{Allen-ZhuL17, CarmonDST19}. For completeness we give a self-contained proof of the following guarantee on Algorithm~\ref{alg:main} in Appendix~\ref{sec:sgd}.

\begin{restatable}{proposition}{restatemainsgd}\label{prop:main_sgd}
Let $\ma \in \R^{m \times n}$ satisfy $\norm{\ma}_{\max} \le 1$ and $\epsilon, \alpha \in (0, 1)$. Let $\delta\le \frac{\epsilon}{20}$, $\eta = \frac \eps {60}$, and $T = \Theta(\eps^{-2} \log \frac{mn}{\alpha})$ for an appropriate constant.
 With probability $\ge 1 - \alpha$,  Algorithm~\ref{alg:main}  outputs an $\eps$-approximate NE for $\ma$.
\end{restatable}

Given \Cref{prop:main_sgd} to obtain our faster zero-sum game solvers, we simply need to efficiently implement the Gibbs sampling in \Cref{line:gibbs_sample}. As introduced in Section~\ref{sec:intro}, \Cref{prob:sample}, describes a  dynamic approximate Gibbs oracle sampling problem sufficient for this task. Indeed, solving two appropriate parameterizations of Problem~\ref{prob:sample} provides the oracles needed by Algorithm~\ref{alg:main}. By combining Proposition~\ref{prop:main_sgd} with the following~\Cref{prop:mainhintmaintain} (our solution to Problem~\ref{prob:sample}, discussed in greater detail in Section~\ref{ssec:hintsample_overview}), we prove our main result Theorem~\ref{thm:quantumgames}.

\begin{restatable}{theorem}{restatemainhintmaintain}\label{prop:mainhintmaintain}
	Let $\alpha \in (0, 1)$ and $\delta \le \eta$. Given a quantum oracle for $\ma \in \R^{m \times n}$ (defined in \Cref{sec:prelims}) with  $\norm{\ma}_{\max} \leq 1$, we can solve Problem~\ref{prob:sample} with probability $\ge 1 - \alpha$ with 
	\[\max(\tsamp, \tup) = O\Par{1+\sqrt{n} \cdot T\eta\log^4\Par{\frac{mn}{\delta}} \cdot \Par{\sqrt{\eta\log\Par{\frac{n\eta T}{\alpha}}} + \eta\log\Par{\frac{n\eta T}{\alpha}}}},\]
	and an additive initialization cost of \[O\Par{\eta^3 T^3\log^4\left(\frac{n\eta T}{\delta}\right)+\log^7\left(\frac{n\eta T}{\delta}\right)}.\]
\end{restatable}

\begin{theorem}\label{thm:quantumgames}
	Let $\ma \in \R^{m \times n}$ satisfy $\norm{\ma}_{\max} \le 1$, and let $\eps, \alpha \in (0, 1)$. Given a quantum oracle for $\ma$ (defined in \Cref{sec:prelims}), there is a quantum algorithm which yields a classical output $(u, v) \in \Delta^m \times \Delta^n$  that is an $\epsilon$-approximate NE for $\ma$ with probability $\ge 1 - \alpha$ in time
	\[O\Par{\frac{\sqrt{m + n}}{\eps^{2.5}} \log^{4}\Par{\frac{mn}{\eps}} \log^{2.5}\Par{\frac{mn}{\alpha\eps}} + \frac{\sqrt{m + n}}{\eps^{2}} \log^{4}\Par{\frac{mn}{\eps}}\log^{3}\Par{\frac{mn}{\alpha\eps}} + \frac 1 {\eps^3} \log^7\Par{\frac{mn}{\eps}} }.\]
\end{theorem}
\begin{proof}
	We apply two instances of~\Cref{prop:mainhintmaintain} to implement the $\delta$-approximate Gibbs oracle for the dynamic vectors $ -\ma^\top x_t$ and $\ma y_t$, to implement each iteration of Algorithm~\ref{alg:main} in amortized $O(1 + \tsamp + \tup)$ time. Using the settings of parameters $T, \eta$ in Proposition~\ref{prop:main_sgd} and setting $\delta = \Theta(\eps)$, which suffices for Algorithm~\ref{alg:main} and~\Cref{prop:mainhintmaintain}, we have
	\[\max(\tsamp, \tup) = O\Par{\frac{\sqrt{m + n}}{ \eps}\log^{4}\Par{\frac{mn}{\eps}} \log\Par{\frac{mn}{\alpha\eps}} \Par{\eps\log\Par{\frac{mn}{\alpha\eps}} + \sqrt{\eps\log\Par{\frac{mn}{\alpha\eps}}}}}.\] 
	The conclusion follows since, by observation, Algorithm~\ref{alg:main} costs $O(T \cdot (1 + \tsamp + \tup))$. As remarked in the introduction, the additive term in the runtime comes from the cost of stably implementing a quantum circuit required in the use of~\Cref{prop:mainhintmaintain} representing a polynomial transformation in finite precision, which we discuss in greater detail in Appendix~\ref{app:stablepoly}.
\end{proof}

\subsection{Dynamic sampling maintenance via dynamic hint maintenance}\label{ssec:hintsample_overview}

In this section, we overview our proof of~\Cref{prop:mainhintmaintain}, which proceeds in two steps.

\begin{enumerate}
	\item We reduce sampling maintenance (Problem~\ref{prob:sample}) to a problem which we call \emph{hint maintenance}. This latter problem is a specialization of the sampling maintenance problem where suitable advice, which we call the \emph{hint} throughout, is provided.
	\item We show how to solve the hint maintenance problem required by Proposition~\ref{prop:mainhintsample} in~\Cref{prop:mainhintmaintain}, by recursively calling Proposition~\ref{prop:mainhintsample} in phases, allowing us to maintain hints of suitable quality.
\end{enumerate}

\paragraph{Reducing sampling maintenance to hint maintenance.} 

First, we introduce the following data structure for maintaining the $x$ variable in Problem~\ref{prob:sample}, which was used crucially in \cite{vanApeldoornG19} for dynamic Gibbs sampling. This data structure allows efficient queries to subsets of the coordinates of $x$ and we use it in our Gibbs sampler as well.

\begin{lemma}[Sampler tree]\label{lem:samplertree}
	Let $\eta \in \R_{\ge 0}$ and $m \in \N$. There is a classical data structure, $\samptree$, supporting a tree on $O(m)$ nodes such that $[m]$ corresponds to leaves, with the following operations.
	\begin{itemize}
		\item $\init(m, \eta_{\textup{fixed}})$: initialize $x \gets \mathbf{0}_m$ and $\eta \gets \eta_{\textup{fixed}}$
		\item $\update(i)$: $x_i \gets x_i + \eta$
		\item $\subsum(v)$: return the sum of all $x_i$, where $i$ is in the subtree of $v$
	\end{itemize}
	The total runtime of $T$ calls to $\update$ is $O(T\log m)$, and calls to $\subsum$ cost $O(1)$.
\end{lemma}

An implementation of $\samptree$ based on propagating subtree sums upon updates is standard classical data structure, and we omit further description for brevity. Next, we state our first building block towards solving Problem~\ref{prob:sample}, a result which can be thought of as quantum sampling with a hint. We defer its proof to Appendix~\ref{sec:quantum}, as it is primarily based on generalizing dynamic block-encoding strategies with bounded-degree polynomial approximations, as pioneered by \cite{GilyenSLW19, vanApeldoornG19}.

\begin{restatable}{proposition}{restatemainhintsample}\label{prop:mainhintsample}
	Let $x \in \R^m_{\ge 0}$ correspond to an instance of $\samptree$, and $\beta \ge \norm{x}_1$. Let $p$ be the Gibbs distribution associated with $\ma^\top x$, let $Z \defeq \sum_{j \in [n]} \exp([\ma^\top x]_j)$ and $\tZ \in [Z, CZ]$ for some $C \ge 1$. Finally, let $q \in \R^n$ have entries classically queriable in $O(1)$ time, satisfy $q \ge p$ entrywise, $q_j \in [\frac \delta n, 1]$ for all $j \in [n]$, and $\norm{q}_1 = \rho$. Suppose $\tZ$, $C$, $\rho$, and $\beta$ are explicitly known. Given a quantum oracle for $\ma \in \R^{m \times n}$ (defined in \Cref{sec:prelims}) with $\norm{\ma}_{\max} \leq 1$, we can implement a $\delta$-approximate Gibbs oracle which has query cost $O(\sqrt{\rho C} \cdot \beta \log^{4}\Par{\frac{Cmn}{\delta}})$. The total additional cost incurred if $x$ undergoes $T$ $\update$ calls which preserve the invariants on $\tZ, C, \rho, \beta$ is $O(T \log m)$.
\end{restatable}

Proposition~\ref{prop:mainhintsample} makes use of an overestimating hint vector $q$ and approximate normalization constant $\tZ$, which we collectively call the \emph{hint}. The acceptance probability of our rejection sampling is governed by two primary parameters: $\rho = \|q\|_1$, which reflects the degree of overestimation (and can be thought of as a hint quality), and $C \ge 1$, which reflects our inability to accept with probability $\frac{p_j}{q_j}$ when $p$ is implicit (which can be thought of as a normalization quality). In particular, the rejection sampling scheme used in Proposition~\ref{prop:mainhintsample} will instead accept with probability $\frac{p_j}{Cq_j}$.\footnote{Exactly computing $Z$ may require time $\Omega(n)$ in standard implementations, an obstacle to runtimes $\propto \sqrt{n}$.}

Here we elaborate briefly on the implementation of Proposition~\ref{prop:mainhintsample} (for more details, see Appendix~\ref{sec:oracle}). We follow notation of Proposition~\ref{prop:mainhintsample}, and also let $w \defeq \ma^\top x$ such that the unnormalized Gibbs distribution is $\exp(w)$, and $p = \frac{\exp(w)}{Z}$. Proposition~\ref{prop:mainhintsample} is a rejection sampler which first loads the hint $q$ into superposition, and then applies a filter. Overall, our scheme has the form
\begin{equation}\label{eq:reject_gen}\text{sample } j \sim \frac{q}{\rho},\text{ then accept with probability } \frac{\exp(w_j)}{CZ \cdot q_j} = \frac{p_j}{Cq_j},\end{equation}
which results in an accepted sample with probability $\approx \frac 1 {\rho C}$, and hence requires $\approx \sqrt{\rho C}$ trials to succeed after applying quantum amplitude amplification, a generalization of Grover search \cite{BrassardHMT02}.\footnote{The $\beta$ in Proposition~\ref{prop:mainhintsample} comes from loading $\exp(w_j)$ into a quantum oracle via polynomials of degree $\approx \beta$.} The latter filtering step is implemented using appropriate block-encoding technology.

The above discussion suggests that the hint and normalization qualities, parameterized by $\rho$ and $C$, are crucial in controlling the acceptance probability of our scheme. More concretely, in our applications of Proposition~\ref{prop:mainhintsample}, $\beta = \eta T = \tO(\frac 1 \eps)$, which is the bound on the $\ell_1$ norm of the $x_t$ and $y_t$ iterates in Algorithm~\ref{alg:main} under the parameter settings of Proposition~\ref{prop:main_sgd}. Overall, the cost of implementing an approximate Gibbs oracle is then (up to logarithmic factors) $\sqrt{\rho C} \cdot \frac 1 \eps$. Proposition~\ref{prop:mainhintsample} hence reduces  Problem~\ref{prob:sample} to the problem of maintaining the hint consisting of a vector $q$ and a normalization estimate $\tZ$. We mention that Proposition~\ref{prop:mainhintsample} is a strict generalization of a corresponding building block in \cite{vanApeldoornG19}, which only used $q$ set to the all-ones vector.

\paragraph{Approaches for Problem~\ref{prob:sample}.} We now overview our improved solution to Problem~\ref{prob:sample} via efficient use of Proposition~\ref{prop:mainhintsample}. To motivate our solution, we outline three solutions to Problem~\ref{prob:sample} offering different tradeoffs in the overall quality $\rho C$. The first only uses classical information and does not use Proposition~\ref{prop:mainhintsample} at all, the second uses Proposition~\ref{prop:mainhintsample} but maintains no history across iterates, and the third (building upon the first two) is our approach.\\
\\
\noindent\textit{Solution 1: }\cite{GrigoriadisK95}. A standard way to solve Problem~\ref{prob:sample} is to explicitly update $w = \ma^\top x$ and $\exp(w)$, and exactly maintain the normalizing constant $Z$. This allows us to sample from $p$ in $\tO(1)$ time. Since $w$ changes by one row of $\ma$ under a $1$-sparse $\update$ operation to $x$, this is implementable in $O(n)$ time per iteration. We can view this as an instance of the scheme \eqref{eq:reject_gen} with $q = p$, $C = 1$, and $\rho = 1$. It yields the (unbalanced) tradeoff for Problem~\ref{prob:sample} of $\tsamp = \tO(1)$ and $\tup = O(n)$.\\
\\
\noindent\textit{Solution 2: }\cite{vanApeldoornG19}. A recent work \cite{vanApeldoornG19}  introduced a quantum implementation of the scheme \eqref{eq:reject_gen} with an improved tradeoff. The \cite{vanApeldoornG19} scheme first uniformly samples, which in the language of \eqref{eq:reject_gen} means $q = \1_n$ and $\rho = n$. It then applies quantum maximum finding \cite{DurrH96} to obtain an approximate maximum entry of $w$, which they show takes time $\tO(\beta \cdot \sqrt n)$; for the sake of simplicity here, we assume this exactly yields $w_{\max} \defeq \max_{j \in [n]} w_j$. Finally, the acceptance probability $\frac{p_j}{Cq_j}$ is set to $\exp(w_j - w_{\max})$. For $q = \1_n$, this translates to 
\[p_j \cdot \exp(w_{\max} - w_j) = \frac{\exp(w_{\max})}{Z} \le 1,\]
implying $C = 1$ suffices. We note this bound on $C$ can be tight when $w$ is very non-uniform. Overall, the \cite{vanApeldoornG19} scheme's update time requires maximum finding, and its sampling time (via Proposition~\ref{prop:mainhintsample}) requires time $\tO(\beta \cdot \sqrt{\rho C}) = \tO(\beta \cdot \sqrt{n} )$. For $\beta = \tO(\frac 1 \eps)$ as in Algorithm~\ref{alg:main}, this yields the balanced tradeoff $\max(\tsamp, \tup) = \tO\Par{\sqrt{n} \cdot \eps^{-1}}$. As discussed earlier, our key insight is to improve upon this specific choice of hint in \cite{vanApeldoornG19}, for their implicit use of Proposition~\ref{prop:mainhintsample}.\\

\noindent\textit{Solution 3: this work.} We design better hints for Proposition~\ref{prop:mainhintsample} by executing our algorithm in phases corresponding to batches of $\approx \frac 1 \eta$ iterations. At the start of each phase, we use the Gibbs access afforded by Proposition~\ref{prop:mainhintsample} to produce a suitable hint for efficiently implementing the next phase. Our execution of this strategy, parameterized by an integer $k \in [n]$, relies on the following observations.

\begin{enumerate}
	\item During $\lceil\frac 1 \eta\rceil$ iterations $t \in \{\tau + s\}_{s \in [\lceil\frac 1 \eta\rceil]}$ (where $\tau$ starts the phase), the dynamic Gibbs distribution $p_t$ (where $t$ is the iteration index) changes by $O(1)$ multiplicatively, since $w$ entrywise changes by $O(1)$ additively. Thus, the quality of a hint vector deteriorates by at most a constant in the phase, so it suffices to give a good hint $q_{\tau} \ge p_{\tau}$ at the phase start.
	\item By using access to Proposition~\ref{prop:mainhintsample} at the end of the previous phase, we can efficiently estimate large entries of $p_{\tau}$. More precisely, we sample $\tO(k)$ times from $p_{\tau}$, and let the empirical distribution of these samples be $\tq$. Chernoff bounds show that any large entry $[p_{\tau}]_j = \Omega(\frac 1 k)$ will be accurately reflected in the empirical sample. Hence, we set the hint to
	\[q_j = \begin{cases}
	\tq_j \cdot O(1) & \tq_j = \Omega(\frac 1 k) \\
	\frac 1 k \cdot O(1) & \tq_j = O(\frac 1 k)
	\end{cases},
		\]
	for appropriate constants. This yields an improved hint quality of $\rho \approx \frac n k$, since large entries of the hint sum to at most $O(1)$ (as $\tq_j \approx p_j$), and small entries sum to $O(\frac n k)$.
	\item We show a similar strategy of using empirical concentration, combined with a testing variant of Proposition~\ref{prop:mainhintsample}, accurately estimates the normalizing factor $Z$, yielding $C = O(1)$.
\end{enumerate}
This strategy yields $\tsamp = \tO(\beta \cdot \sqrt{n/k})$ and $\tup = \tO(\tsamp \cdot k\eta)$ (since we amortize $\tup$ over $\approx \frac 1 \eta$ iterations). For the parameter settings of Algorithm~\ref{alg:main}, optimizing $k$ yields
\[
\max(\tsamp, \tup) = \tO\Par{\sqrt{n} \cdot \eps^{-\half}}\,.
\]

We prove~\Cref{prop:mainhintmaintain}, our improved solution to Problem~\ref{prob:sample}, in Section~\ref{sec:oracle}. Ignoring logarithmic factors and assuming $\eta \ll 1$ (as in our setting),~\Cref{prop:mainhintmaintain} shows we can maintain $\max(\tsamp, \tup) = \tO(\sqrt{n} \cdot T\eta^{1.5})$. For the parameter settings $T = \tO(\eps^{-2})$, $\eta = \Theta(\eps)$, as stated in Proposition~\ref{prop:main_sgd}, this indeed equates to $\max(\tsamp, \tup) = \tO(\sqrt{n} \cdot \eps^{-\half})$. %

\section{Gibbs sampling oracle implementation}
\label{sec:oracle}

In this section, we prove~\Cref{prop:mainhintmaintain}, which gives our solution to Problem~\ref{prob:sample}. To do so, we follow the outline given in Section~\ref{ssec:hintsample_overview}, wherein we solve Problem~\ref{prob:sample} in batches of $\lceil\frac 1 \eta\rceil$ iterations, each of which we call a ``phase.'' In Sections~\ref{ssec:initial} and~\ref{ssec:eachiter}, we only discuss a single phase of Problem~\ref{prob:sample}, consisting of the iterations $\tau + s$ for $s \in [\lceil\frac 1 \eta\rceil]$ and some initial iteration $\tau$, assuming certain invariants (stated below) hold at the start of the phase. We give a complete solution to Problem~\ref{prob:sample} in Section~\ref{ssec:solveproblem}.

\begin{invariant}[Approximate normalization access]\label{inv:approxz}
	We explicitly have $\tZp$ with $\tZp \in [Z_\tau, CZ_\tau]$ for some $C = O(1)$.
\end{invariant}

\begin{invariant}[Initial sampling maintenance]\label{inv:samplesolve}
	We have $\oracle_\tau$ solving Problem~\ref{prob:sample} in iteration $\tau$.\label{item:inv2}
\end{invariant}

The remainder of this section is then organized as follows.

\begin{itemize}
	\item Section~\ref{ssec:initial}: We show that assuming Invariants~\ref{inv:approxz} and~\ref{inv:samplesolve} hold at the start of a phase, we can perform preprocessing used to construct our hint, consisting of the estimated normalization $\tZ$ and vector $q$, in an application of Proposition~\ref{prop:mainhintsample}. This gives the cost of $\tsamp$ in Problem~\ref{prob:sample}. 
	\item Section~\ref{ssec:eachiter}: We show that at the conclusion of each phase we can maintain Invariants~\ref{inv:approxz} and~\ref{inv:samplesolve} for use in the next phase. This gives the cost of $\tup$ in Problem~\ref{prob:sample}.
	\item Section~\ref{ssec:solveproblem}: We recursively call the subroutine of Sections~\ref{ssec:initial} and~\ref{ssec:eachiter} (which solves Problem~\ref{prob:sample} for all the iterations $\tau + s$ where $s \in [\lceil\frac 1 \eta\rceil]$ for some $\tau$) $\approx \eta T$ times to prove~\Cref{prop:mainhintmaintain}.
\end{itemize}

\subsection{Preprocessing and approximate Gibbs oracle implementation}\label{ssec:initial}

In this section, we show how to construct the ``hint'' $q$ which will be used throughout a phase (starting in iteration $\tau$) given access to $\oracle_\tau$, and bound $\rho = \norm{q}_1$ which quantifies the quality of our hint, under the assumption that Invariants~\ref{inv:approxz} and~\ref{inv:samplesolve} hold in the phase. We first show a multiplicative stability property of the relevant Gibbs distributions in a phase.

\begin{lemma}\label{lem:stableiter}
For all $s \in [\lceil\frac 1 \eta\rceil]$, we have
\[Z_{\tau + s} \in \Brack{\third Z_\tau, 3Z_\tau},\text{ and } p_{\tau + s} \in \Brack{\ninth p_\tau, 9p_\tau} \text{ entrywise.}\]
\end{lemma}
\begin{proof}
Let $\nu_t \defeq \exp(\ma^\top x_t)$ for all $t$, such that $p_t = \frac{\nu_t}{Z_t}$. We have that for any $j \in [n]$,
\begin{align*}\frac{[\nu_{\tau + s}]_j}{[\nu_\tau]_j} &= \exp\Par{\Brack{\ma^\top \Par{x_{\tau + s} - x_\tau}}_j} \\
	&\in\Brack{\exp\Par{-\norm{\ma}_{\max} \norm{x_{\tau + s} - x_\tau}_1} , \exp\Par{\norm{\ma}_{\max} \norm{x_{\tau + s} - x_\tau}_1} } \\
	&\in\Brack{\exp\Par{-\eta s}, \exp\Par{\eta s}} \in \Brack{\third, 3}.\end{align*}
Similarly, $Z_{\tau + s} \in [\third Z_\tau, 3Z_\tau]$, and combining yields the conclusion.
\end{proof}

Next, our computation of the overestimating vector $q$ is parameterized by an integer $k \in [n]$ which will be fixed throughout this section and Section~\ref{ssec:eachiter}. We will simply set $q$ to be an upscaled variant of an empirical distribution of roughly $k$ draws from $\oracle_\tau$.

\begin{lemma}\label{lem:qgood}
Let $k \in [n]$, $\alpha \in (0, 1)$, and suppose $\delta \le \frac{1}{16k}$. Draw $N = \Theta(k \log \frac {n\eta T} \alpha)$ samples from $\oracle_\tau$ for an appropriately large constant, and let $\tq \in \Delta^n$ be the empirical distribution over these $N$ samples. Define $\bigset \defeq \{i \in [n] \mid \tq_i \ge \frac{1}{2k}\}$. Then for
\[q_j = \begin{cases}
18\tq_j	& j \in \bigset \\
\frac {18} k & j \not\in \bigset
\end{cases},\]
with probability $\ge 1 - \frac{\alpha}{2\lceil\eta T\rceil}$, $\norm{q}_1 = O(\frac n k)$ and $q \ge p_{\tau + s}$ entrywise, for all $s \le \frac 1 \eta$.
\end{lemma}
\begin{proof}
The first conclusion $\norm{q}_1 = O(\frac n k)$ is immediate from the definition of $q$, since $\norm{q}_1 \le 18\norm{\tq}_1 + \frac{18n}{k}$. In light of Lemma~\ref{lem:stableiter} (which holds deterministically), to show the second conclusion, it suffices to show that with the desired success probability, we have both
\begin{equation}\label{eq:bigsetgood}
2\tq_j \ge [p_\tau]_j \text{ for all } j \in \bigset
\end{equation}
and
\begin{equation}\label{eq:smallsetgood}
\frac 2 k \ge [p_\tau]_j \text{ for all } j \not\in \bigset.
\end{equation}
Denote $\alpha' \defeq \frac \alpha {2\lceil\eta T\rceil}$ for notational convenience, and let $\tp$ denote the distribution of samples from $\oracle_\tau$, and recall that $\norm{\tp - p_\tau}_1 \le \frac 1 {16k}$. Because we are taking $\Theta(k \log \frac n {\alpha'})$ samples from $\tp$, we have by a standard Chernoff bound that with probability at least $1 - \alpha'$ (union bounding over all coordinates $j \in [n]$), both of the following hold.
\begin{enumerate}
	\item For all $j \in [n]$ such that $\tp_j \ge \frac 1 {4k}$, $\tq_j \ge \frac {2\tp_j}{3}$.
	\item For all $j \in [n]$ such that $\tp_j \le \frac 1 {4k}$, $\tq_j \le \frac 1 {2k}$.
\end{enumerate}
We condition on these events for the remainder of the proof; we now show \eqref{eq:bigsetgood}, \eqref{eq:smallsetgood} in turn. \\

\textit{Proof of \eqref{eq:bigsetgood}.} To see \eqref{eq:bigsetgood}, the second event above implies that if $\tp_j \le \frac{1}{4k}$, then $j\not\in \bigset$. Hence, for all $j \in \bigset$, we have $\tq_j \ge \frac{2\tp_j}{3} \ge \frac{[p_\tau]_j}{2}$ since $\norm{\tp - p_\tau}_\infty \le \frac 1 {16k} \le \frac 1 4 \tp_j $ for all $j \in \bigset$. \\

\textit{Proof of \eqref{eq:smallsetgood}.} To see \eqref{eq:smallsetgood}, suppose for contradiction that $j \not\in \bigset$ and $[p_\tau]_j > \frac 2 k$. This implies that $\tp_j > \frac 1 k$, and hence by the first event above, $\tq_j \ge \frac 1 {2k}$, contradicting $j \not\in \bigset$.
\end{proof}

\begin{corollary}\label{cor:phasesolve}
Assume that Invariants~\ref{inv:approxz},~\ref{inv:samplesolve} hold for the phase consisting of iterations $\tau + s$, $s \in [\lceil\frac 1 \eta\rceil]$. We can solve Problem~\ref{prob:sample} for the phase with probability $\ge 1 - \frac{\alpha}{2\lceil\eta T\rceil}$, and
\[\tsamp \defeq O\Par{\sqrt{\frac{n}{k}} \cdot T\eta \log^{4}\Par{\frac{mn}{\delta}}}.\]
\end{corollary}
\begin{proof}
We will run the algorithm described in the proof of Lemma~\ref{lem:qgood}, and condition on it succeeding, giving the failure probability. It then suffices to apply Proposition~\ref{prop:mainhintsample} with $q$ defined in Lemma~\ref{lem:qgood}. For this $q$, we parameterize Proposition~\ref{prop:mainhintsample} with $C = O(1)$ (see Invariant~\ref{inv:approxz}), $\rho = O(\frac n k)$ (see Lemma~\ref{lem:qgood}), and $\beta = T\eta$. It is clear the lower bound on entries of $q$ in Proposition~\ref{prop:mainhintsample} holds.
\end{proof}

\subsection{Maintaining invariants}\label{ssec:eachiter}

We now show how to maintain Invariant~\ref{inv:approxz} at iteration $\tau' \defeq \tau + \lceil\frac 1 \eta\rceil$, for use in the next phase, and bound the cost of doing so. We note that Invariant~\ref{inv:samplesolve} follows immediately from our construction in Corollary~\ref{cor:phasesolve}. First, by combining Lemma~\ref{lem:stableiter} with Invariant~\ref{inv:approxz}, 
\begin{equation}Z_{\tau'} \in \Brack{\frac {\tZp} {3C}, 3\tZp}.\end{equation}
This suggests that we may use $3\tZp = \tZ$ for the next phase; however, this would lead to an exponential blowup in the multiplicative range $C$. To sidestep this, we develop a tester for a hidden parameter governing a success probability, which will be used to give a refined estimate $\tZ$. We require the following corollary of Proposition~\ref{prop:mainhintsample}, whose proof we defer to Appendix~\ref{sec:quantum}.

\begin{restatable}{corollary}{restateotest}\label{cor:otest}
	Following notation of Proposition~\ref{prop:mainhintsample}, let $R \defeq \frac{\tZ}{Z}$. There is a quantum oracle $\otest$ which can be implemented under $T$ $\update$ calls to $x$ in $O(T \log m)$ time, and has query cost
	\[O\Par{\sqrt{\rho C} \cdot \beta\log^{4} \Par{\frac{Cmn}{\ell\delta}}}.\]
	Furthermore, for explicitly known constants $C_\ell$ and $C_u$, $\otest$ returns ``success'' with probability $p$ for
	\[\frac{C_\ell}{\sqrt{R\rho}} \le p \le \frac{C_u}{\sqrt{R\rho}}.\]
\end{restatable}

Corollary~\ref{cor:otest} differs from Proposition~\ref{prop:mainhintsample} in that it returns a Boolean-valued answer (as opposed to a sample from an approximate Gibbs distribution), and has a success probability parameterized by explicit constants. We now show how to use Corollary~\ref{cor:otest} to maintain Invariant~\ref{inv:approxz}. 

\begin{lemma}\label{lem:maintainZ}
Assume Invariants~\ref{inv:approxz},~\ref{inv:samplesolve} hold for the phase consisting of iterations $\tau + s$, $s \in [\lceil\frac 1 \eta\rceil]$, and suppose $C \ge \frac{4C_u^2}{C_\ell^2}$ for $C = O(1)$, where $C_u$ and $C_\ell$ are the constants from Corollary~\ref{cor:otest}. Further, suppose we have obtained $q$ satisfying the conclusion of Lemma~\ref{lem:qgood} (i.e.\ that the algorithm in Lemma~\ref{lem:qgood} succeeded). We can determine $\tZ$ such that 
$\tZ \in [Z_{\tau'}, CZ_{\tau'}]$ with probability $\ge 1 - \frac{\alpha}{2\lceil\eta T\rceil}$, in time
\[O\Par{\sqrt{\frac{n}{k}} \cdot T\eta \log^{4}\Par{\frac{mn}{\delta}} \log\Par{\frac{\eta T}{\alpha}}}.\]
\end{lemma}
\begin{proof}
Define $\tZ_0 \defeq 3\tZp$, $R_0 \defeq \frac{\tZ_0}{Z_{\tau'}}$, and note that $\tZ_0 \in [Z_{\tau'}, 9CZ_{\tau'}]$ by Invariant~\ref{inv:approxz} and Lemma~\ref{lem:stableiter}. Next, assuming the success of Lemma~\ref{lem:qgood}, we have that the success probability $p$ of $\otest$ from Corollary~\ref{cor:otest} using the estimate $\tZ_0$ satisfies (for the unknown $R_0 \in [1, 9C]$, and known $C_\ell, C_u, \rho$)
\[\frac{C_\ell}{\sqrt{R_0\rho}} \le p \le \frac{C_u}{\sqrt{R_0\rho}}.\]
For $N \defeq 27\log \frac{4\lceil\eta T\rceil}{\alpha} \cdot \frac{3\sqrt{C\rho}}{C_\ell}$, we first run $\otest$ $N$ times and check the number of successes, denoted by $S$, which fits within the runtime budget by Corollary~\ref{cor:otest}. By a Chernoff bound, we have that with probability $\ge 1 - \frac{\alpha}{2\lceil\eta T\rceil}$, we have
\[54\log \frac{4\lceil\eta T\rceil}{\alpha} \cdot \sqrt{\frac{C}{R_0}} \le \frac 2 3 pN \le S \le \frac 4 3 pN \le 108\log\frac{4\lceil\eta T\rceil}{\alpha} \cdot \frac{C_u}{C_{\ell}} \cdot \sqrt{\frac C {R_0}}.\]
Hence, we can determine the quantity $R_0$ up to a multiplicative factor of $\frac{4C_u^2}{C_\ell^2} \le C$, which also implies the same multiplicative approximation factor for $Z_{\tau'}$, as desired.
\end{proof}

\subsection{Proof of~\Cref{prop:mainhintmaintain}}\label{ssec:solveproblem}

\restatemainhintmaintain*
\begin{proof}
We first claim that for any $k \in [n]$, we can solve Problem~\ref{prob:sample} with probability $\ge 1 - \alpha$ and
\begin{align*}
	\tsamp &= O\Par{\sqrt{\frac{n}{k}} \cdot T\eta \log^{4}\Par{\frac{mn}{\delta}}}, \\
	\tup &= O\Par{\Par{\sqrt{\frac{n}{k}} \cdot T\eta \log^{4}\Par{\frac{mn}{\delta}}} \cdot k\eta \log\Par{\frac{n\eta T}{\alpha}}}.
\end{align*}
This follows from combining Lemma~\ref{lem:qgood} (amortized over $\lceil\frac 1 \eta\rceil$ iterations), Corollary~\ref{cor:phasesolve}, and Lemma~\ref{lem:maintainZ}, and taking a union bound over at most $\lceil\eta T\rceil$ phases. Here we note that the cost of $\log m$ per iteration to support $\update$ costs to $x$ in Lemma~\ref{lem:samplertree}, Proposition~\ref{prop:mainhintsample}, and Corollary~\ref{cor:otest} is not dominant. By choosing $k = \Theta(\max(1, (\eta \log \frac{mn}{\alpha\eps})^{-1}))$, we balance the costs of $\tsamp$ and $\tup$, yielding the conclusion. We finally note that by picking an appropriate constant in the definition of $k$, we have $\delta \le \eta \implies \delta \le \frac{1}{16k}$ as required by Lemma~\ref{lem:qgood}, the only component specifying a bound on $\delta$.
\end{proof} 
\subsection*{Acknowledgments}

We thank Andr\'as Gily\'en for communication regarding the prior work \cite{vanApeldoornG19}. AB was supported in part by the DOE QuantISED grant DE-SC0020360, by the AFOSR under grant FA9550-21-1-0392, and by the U.S. DOE Office of Science under Award Number DE-SC0020266. YG was supported in part by the Stanford MS\&E DE\&I Research program. YJ was supported in part by a Stanford Graduate Fellowship and a Danzig-Lieberman Graduate Fellowship. AS was supported in part by a Microsoft Research Faculty Fellowship, NSF CAREER Award CCF1844855, NSF Grant CCF-1955039, a PayPal research award, and a Sloan Research Fellowship. KT thanks Ewin Tang for her expertise on quantum linear algebra and for fielding many of our questions.

\bibliographystyle{alpha}	
\newcommand{\etalchar}[1]{$^{#1}$}

\newpage
\begin{appendix}

\section{Solving matrix games with a Gibbs sampling oracle}
\label{sec:sgd}

In this section, we prove Proposition~\ref{prop:main_sgd}, which shows how to solve a zero-sum matrix game using an approximate Gibbs sampling oracle (via Algorithm~\ref{alg:main}). To briefly motivate the algorithm we use and our proof of its guarantees, we recall the problem we consider is of the form
\begin{equation}\label{def:l1-l1}
\begin{gathered}
\min_{v\in\Delta^n}\max_{u\in\Delta^m} f(u,v)\defeq u^\top \ma v,~~\text{where}~~\norm{\ma}_{\max} \le 1,
\end{gathered}
\end{equation}
and we define the associated gradient operator as
\begin{align}\label{def:l1-l1-grad-operator}
	g(u,v) = (-\ma v, \ma^\top u).
\end{align}
Taking (stochastic) mirror descent steps on the gradient operator in \eqref{def:l1-l1} is well-known to yield an approximate NE to the matrix game \cite{Bubeck15}. We show that an approximate implementation of this strategy, combined with appropriate subsampling, efficiently yields an approximate NE. We begin by making the following observation.
\begin{lemma}\label{lem:deltabias}
Let $u,  \tu \in \Delta^m$ have $\norm{u - \tu}_1 \le \delta$. Let $\tg \defeq \ai$ where $i \sim \tu$, and $g \defeq \ma^\top u$. Then, $\norm{g - \E \tg}_\infty \le \delta$.
\end{lemma}
\begin{proof}
Note that $\E \tg = \ma^\top \tu$, and $\norm{\ma^\top(u - \tu)}_\infty \le \norm{u - \tu}_1 \le \delta$ since $\norm{\ma}_{\max} \le 1$.
\end{proof}

We next present a variant of the classical mirror descent analysis, which bounds the expected approximation quality of iterates of Algorithm~\ref{alg:main} prior to subsampling.

\begin{proposition}\label{prop:main-exp}
Let $\delta\le \frac{\epsilon}{20}$, $\eta = \frac{\eps}{15}$ and $T \ge \frac{6\log(mn)}{\eta\eps}$ in Algorithm~\ref{alg:main}. Let the iterates of Algorithm~\ref{alg:main} be $\{x_t, y_t\}_{t = 0}^{T - 1}$, and denote $u_t \defeq \frac{\exp(\ma y_t)}{\norm{\exp(\ma y_t)}_1}$, $v_t \defeq \frac{\exp(-\ma^\top x_t)}{\norm{\exp(-\ma^\top x_t)}_1}$ for all $0 \le t < T$. For $(\bu, \bv) \defeq \frac{1}{T}\sum_{t = 0}^{T - 1}(u_t, v_t)$, we have
\begin{align}\label{eq:main-exp}
\E\left[\max_{u \in \Delta^m} u^\top \ma \bv  - \min_{v \in \Delta^n} \bu^\top \ma v \right]\le \eps.
\end{align}
\end{proposition}

\begin{proof}
By definition of the updates, at every iteration $0 \le t \le T - 1$,  we have
\begin{align*}
u_{t+1} & = \argmin_{u\in\Delta^m}\Brace{\eta \langle -\ma_{:j_t}, u\rangle + \sum_{i \in [m]} [u]_{i}\log\frac{[u]_{i}}{[u_t]_{i}}	},\\
v_{t+1} & = \argmin_{v\in\Delta^n}\Brace{\eta \langle \ma_{i_t:}, v\rangle + \sum_{j \in [n]} [v]_{j}\log\frac{[v]_{j}}{[v_t]_{j}}	}.
\end{align*}
Consequently, by the optimality conditions of $u_{t+1}$ and $v_{t+1}$ respectively, we have for any $u\in\Delta^m$, $v\in\Delta^n$, and letting $V_{x}(x') \defeq  \sum_{k}[x']_k\log \frac{[x']_k}{[x]_k}$ be the KL divergence between simplex variables of appropriate dimension,
\begin{equation}\label{eq:rearrangeopt}
\begin{aligned}
\inprod{-\aj}{u_t - u} + \inprod{\ai}{v_t - v} &\le \frac 1 \eta \Par{V_{u_t}(u) - V_{u_{t + 1}}(u) + V_{v_t}(v) - V_{v_{t + 1}}(v)} \\
&+ \Par{\inprod{-\aj}{u_t - u_{t + 1}}-\frac{1}{\eta}V_{u_t}(u_{t+1})} \\
&+\Par{ \inprod{\ai}{v_t - v_{t + 1}}-\frac{1}{\eta}V_{v_t}(v_{t+1})}\\
& \le \frac 1 \eta \Par{V_{u_t}(u) - V_{u_{t + 1}}(u) + V_{v_t}(v) - V_{v_{t + 1}}(v)} \\
&+ \frac{\eta}{2} \norm{\aj}_\infty^2+ \frac{\eta}{2} \norm{\ai}_\infty^2,
\end{aligned}
\end{equation}
where for the last inequality we use H\"older's inequality and the fact that $V$ is $1$-strongly convex in the $\ell_1$ norm (by Pinsker's inequality). Averaging the above for $0 \le t < T$, and denoting $w_t \defeq (u_t, v_t)$ and $\tilde{g}_t \defeq (-\ma_{:j_t},\ma_{i_t:})$, we obtain for any $w = (u, v)\in\Delta^m\times\Delta^n$,
\begin{equation}\label{eq:regret-1}
\frac{1}{T}\sum_{t = 0}^{T - 1}\langle\tilde{g}_t, w_t-w\rangle\le \frac{1}{\eta T} \Par{V_{u_0}(u)+V_{v_0}(v)}+\eta.
\end{equation}
In the above, we further recalled the bound $\norm{\ma}_{\max} \le 1$ by assumption. In order to bound the deviation of the left-hand side from its expectation, we use a ``ghost iterate'' argument following \cite{NemirovskiJLS09, CarmonJST19}. In particular, we define iterates $\tu_t$, $\tv_t$ as follows: let $\tu_0 \gets u_0$, $\tv_0 \gets v_0$, and then for each $0 \le t < T$, define
\begin{align*}
\tu_{t+1} &\defeq \argmin_{u \in\Delta^m}\Brace{\eta \langle -\ma v_t + \ma_{:j_t}, \bar{u}\rangle + \sum_{i \in [m]} [u]_{i}\log\frac{[u]_{i}}{[\tu_t]_{i}}	},\\
\tv_{t + 1} &\defeq \argmin_{v \in\Delta^n}\Brace{\eta \langle \ma^\top u_t - \ma_{:i_t}, \bar{v}\rangle + \sum_{j \in [n]} [v]_{j}\log\frac{[v]_{j}}{[\tv_t]_{j}}	},
\end{align*}
where $i, j$ above are the same coordinates as were used in defining the updates to $u_{t + 1}$ and $v_{t + 1}$. By an analogous bound to \eqref{eq:rearrangeopt}, where we note $\norm{\ma_{:j_t} - \ma^\top v_t}_\infty, \norm{\ma u_t - \ma_{i_t:}}_\infty \le 2$,
\begin{align*}
\inprod{-\ma^\top v_t + \ma_{:j_t}}{\tu_t - u} + \inprod{\ma u_t - \ma_{i_t:}}{\tv_t - v}
&\le \frac{1}{\eta}\Par{V_{\tu_t}(u)- V_{\tu_{t+1}}(u)+ V_{\tv_t}(v)- V_{\tv_{t+1}}(v)}\\
&+ 4\eta.
\end{align*}
Averaging the above for $0 \le t < T$, and denoting $\tw_t \defeq (\tu_t, \tv_t)$ and $g_t \defeq g(w_t)$ (see \eqref{def:l1-l1}), we obtain for any $w = (u, v)\in\Delta^m\times\Delta^n$, 
\begin{equation}\label{eq:regret-2}
\frac{1}{T}\sum_{t\in[T]-1}\langle g_t - \tilde{g}_t, \tw_t-w\rangle\le \frac{1}{\eta T} \Par{V_{u_0}(u)+V_{v_0}(v)} + 4\eta.
\end{equation}

Summing inequalities~\eqref{eq:regret-1} and~\eqref{eq:regret-2}, and maximizing over $w = (u, v)\in\Delta^m\times \Delta^n$, we have
\begin{equation}\label{eq:regret-expectation}
\begin{aligned}
\max_{w \in\Delta^m\times \Delta^n} \frac{1}{T}\sum_{t = 0}^{T - 1}\langle g_t, w_t-w\rangle 
	&  \le \max_{u\in\Delta^n,v\in\Delta^m}\frac{2}{\eta T}\left(V_{u_0}(u)+V_{v_0}(v)\right)\\
	&+5\eta+\frac{1}{T}\sum_{t=0}^{T - 1}\langle g_t-\tilde{g}_t, w_t-\tw_t\rangle.
\end{aligned}
\end{equation}
Taking expectations over the above, we have
\begin{align*}
	\E\left[\max_{w\in\Delta^m\times \Delta^n}\frac{1}{T}\sum_{t=0}^{T - 1}\langle g_t, w_t-w\rangle\right] &\le \max_{u\in\Delta^n,v\in\Delta^m}\frac{2}{\eta T}\left[V_{u_0}(u)+V_{v_0}(v)\right] \\
	&+5\eta +\E\left[\frac{1}{T}\sum_{t=0}^{T - 1}\langle g_t-\tilde{g}_t, w_t-\tw_t\rangle\right]\\
	 &\stackrel{(i)}{\le}  \frac{2\log (mn)}{\eta T}+5\eta +\frac{1}{T}\sum_{t\in[T]-1}\langle g_t-\E\tilde{g}_t, w_t-\bar{w}_t\rangle,\\
	 &\stackrel{(ii)}{\le} \frac{2\log(mn)}{\eta T}+5\eta+4\delta \stackrel{(iii)}{\le} \epsilon.
\end{align*}
In the above, $(i)$ used the diameter bound of the KL divergence from the uniform distribution, i.e.\ $\max_{u \in\Delta^m} V_{u_0}(u) = \log m$ (and a similar bound for $V_{v_0}(v)$). Further, $(ii)$ uses that $\tilde{g}_t$ is  conditionally independent of $w_t$ and $\tw_t$, and by the assumption on the Gibbs sampler $\norm{g_t - \E \tilde{g}_t}_\infty \le \delta$ (via Lemma~\ref{lem:deltabias}), and H\"older, and $(iii)$ uses our choices of $T$, $\eta$ and $\delta$. 

Finally, we note that the desired claim follows by linearity: for any $w = (u, v)$,
\begin{align*}\frac{1}{T}\sum_{t=0}^{T - 1}\langle g_t, w_t-w\rangle & = \left\langle g\Par{\frac{1}{T}\sum_{t=0}^{T - 1}w_t},\frac{1}{T}\sum_{t=0}^{T - 1}w_t - w\right\rangle\\
& 	
 = u^\top \ma \bv - \bu^\top \ma v.
 \end{align*}
\end{proof}

By using a simple martingale argument (inspired by those in \cite{Allen-ZhuL17, CarmonDST19}) to bound the error term in \eqref{eq:regret-expectation}, we show that the guarantee of Proposition~\ref{prop:main-exp} holds with high probability.

\begin{corollary}\label{prop:main-hp}
Let $\alpha \in (0, 1)$, and let  $\delta\le \frac{\epsilon}{20}$, $\eta = \frac \eps {20}$ and $T \ge \frac{8\log(mn)}{\eta\epsilon}+\frac{2048\log\frac 1 \alpha}{\epsilon^2} $ in Algorithm~\ref{alg:main}. Then with probability at least $1 - \alpha$, following notation of Proposition~\ref{prop:main-exp}, $(\bu, \bv)$ are an $\eps$-approximate NE for $\ma$.

\end{corollary}
\begin{proof}
Consider the filtration given by $\calF_t = \sigma(u_0,v_0,\tilde{g}_0,\cdots, \tilde{g}_{t},u_{t+1},v_{t+1})$. We will bound the terms $\sum_{t=0}^{T - 1}\langle g_t-\tilde{g}_t, w_t-\bar{w}_t\rangle$ in \eqref{eq:main-exp}. To do so, we define a martingale difference sequence of the form $D_t \defeq \langle g_{t}-\tilde{g}_{t}, w_{t}-\bar{w}_{t} \rangle - \langle g_{t}-\E\left[\tilde{g}_{t}|\calF_{t-1}\right], w_{t}-\bar{w}_{t} \rangle$ which is adapted to the filtration $\calF_{t}$. We first note that $D_t\le \norm{g_{t-1}-\tilde{g}_{t-1}}_\infty\norm{w_{t-1}-\bar{w}_{t-1} }_1\le 8$ with probability $1$. Consequently, applying the Azuma-Hoeffding inequality yields
\begin{align*}
	\sum_{t=0}^{T - 1}D_t\le \sqrt{128T\log\frac{1}{\alpha}}~\text{with probability}\ge 1-\alpha.
\end{align*}
Plugging this back into~\eqref{eq:regret-expectation} and using the KL divergence range bound, Lemma~\ref{lem:deltabias} with our definition of $\ogibbs$, and choices of parameters, we thus have with probability $1-\alpha$,
\begin{align}\label{eq:regret-expectation-hp}
	& \max_{w\in\Delta^m\times \Delta^n} \frac{1}{T}\sum_{t=0}^{T - 1}\langle g_t, w_t-w\rangle \le \frac{2\log mn}{\eta T}+5\eta+4\delta+\sqrt{\frac{128\log\frac{1}{\alpha}}{T}}\le\epsilon.
\end{align}
The remainder of the proof follows analogously to Proposition~\ref{prop:main-exp}.
\end{proof}

The Gibbs sampling oracles implicitly maintain access to $u_t \propto \exp(\ma  y_t)$ and $v_t \propto \exp(-\ma^\top x_t)$, which by averaging gives $(\bar{u},\bar{v})=\frac{1}{T}\sum_{t=0}^{T-1}(u_t, v_t)$ as one approximate equilibrium as guaranteed in~\Cref{prop:main-hp}. To turn the implicitly maintained iterates into an actual classic output, we subsample the iterates. Below we formally show one can take the empirical average of independent samples from distributions close to $\bu$ and $\bv$ to also obtain an approximate equilibrium (with the same approximation factor up to constant factors) with high probability.

\begin{lemma}\label{lem:empirical-sample}
Suppose $\bu = \frac 1 T \sum_{t = 0}^{T - 1} u_t$ for $\{u_t\}_{t=0}^{T - 1} \subset \Delta^m$ and $\bv = \frac 1 T \sum_{t = 0}^{T - 1} v_t$ for $\{v_t\}_{t=0}^{T - 1} \subset \Delta^n$ are an $\eps$-approximate NE for $\ma$. Further suppose that for some $\delta \in (0, 1)$, $\{\tu_t\}_{t=0}^{T - 1} \subset \Delta^m$, $\{\tv_t\}_{t=0}^{T - 1} \subset \Delta^n$, and for all $0 \le t < T - 1$, we have $\norm{\tu_t - u_t}_1 \le \delta$ and $\norm{\tv_t - v_t}_1 \le \delta$. Let $\hu = \frac 1 T \sum_{t = 0}^{T - 1} e_{i_t}$ where each $e_{i_t} \in \R^m$ is sampled independently according to $\tu_t$; similarly, let $\hv = \frac 1 T \sum_{t = 0}^{T - 1} e_{j_t}$ where each $e_{j_t} \in \R^n$ is sampled independently according to $\tv_t$. Suppose $T \ge\frac{16\log \frac{mn}{\alpha}}{\eps^2}$. Then with probability at least $1 - \alpha$, $(\hu, \hv)$ are a $(2\eps + 2\delta)$-approximate NE for $\ma$.
\end{lemma}
\begin{proof}
First, let $\tu_{\textup{avg}} = \frac 1 T \sum_{t = 0}^{T - 1} \tu_t$ and $\tv_{\textup{avg}} = \frac 1 T \sum_{t = 0}^{T - 1} \tv_t$. By convexity of norms, we have $\norm{\tu_{\textup{avg}} - \bu}_1 \le \delta$ and $\norm{\tv_{\textup{avg}} - \bv}_1 \le \delta$, and hence under the NE approximation guarantee of $(\bu, \bv)$ and H\"older's inequality,
\[\max_{u \in \Delta^m} u^\top \ma \tv_{\textup{avg}} - \min_{v \in \Delta^m} \tu_{\textup{avg}}^\top \ma v \le \eps + 2\delta.\]
Let $z$ be a fixed vector in $[-1, 1]^n$. By Hoeffding's inequality, since each random variable $\inprod{z}{e_{j_t}}$ lies in the range $[-1, 1]$ and $\E \hv = \tv_{\textup{avg}}$, we have that
\begin{equation}\label{eq:unionbound}\Pr\Brack{\Abs{\inprod{z}{\hv - \tv_{\textup{avg}}}} \ge \frac \eps 2} \le 2\exp\Par{-\frac{T\eps^2}{8}} \le \frac{\alpha}{m + n}.\end{equation}
Next, note that $\max_{u \in \Delta^m} u^\top \ma \tv_{\textup{avg}}$ is achieved by a basis vector $u = e_i$. Hence, applying a union bound over \eqref{eq:unionbound} for all $z = \ai$ shows that with probability at least $1 - \frac{\alpha m}{m + n}$,
\[\max_{u \in \Delta^m} u^\top \ma \hv \le \max_{u \in \Delta^m} u^\top \ma \tv_{\textup{avg}} + \frac \eps 2.\]
By symmetry, with probability at least $1 - \frac{\alpha n}{m + n}$,
\[\min_{v \in \Delta^n} \hu^\top \ma v \ge \min_{v \in \Delta^n} \tu_{\textup{avg}}^\top \ma v - \frac \eps 2.\]
The conclusion follows from a union bound, and combining the above three displays.
\end{proof}

Finally, we put these pieces together to give a complete guarantee.

\restatemainsgd*
\begin{proof}
We follow notation of Proposition~\ref{prop:main-exp}. By applying Corollary~\ref{prop:main-hp} (up to constant factors), we have that with probability at least $1 - \frac \alpha 2$, $\bu \defeq \frac 1 T \sum_{t = 0}^{T - 1} u_t$ and $\bv \defeq \frac 1 T \sum_{t = 0}^{T - 1} v_t$ satisfy
\[\max_{u \in \Delta^m} u^\top \ma \bv - \min_{v \in \Delta^n} \bu^\top \ma v \le \frac \eps 3.\]
Finally, Lemma~\ref{lem:empirical-sample} (with failure probability $\frac \alpha 2$) and a union bound yields the desired conclusion.
\end{proof}

\section{Quantum rejection sampling with a hint}\label{sec:quantum}

In this section, we prove Proposition~\ref{prop:mainhintsample}, which gives a dynamic quantum rejection sampling subroutine and bounds its cost of implementation. Our result is an extension of analogous developments in \cite{vanApeldoornG19}, but are stated more generally to allow for the use of an appropriate ``hint'' vector in the rejection sampling procedure. We build up to our main result in several pieces. 

\paragraph{Amplitude amplification.} First, for a quantum decision algorithm which applies unitary $\mmu$ and then measures, yielding an accepting state with probability $\alpha$, quantum amplification \cite{BrassardHMT02} shows we can apply $\mmu$ $\approx \alpha^{-\half}$ times to obtain an accepting state with high probability.

\begin{proposition}[Theorem 3, \cite{BrassardHMT02}]\label{prop:amplif}Let $S \subseteq \{0, 1\}^s$, let $\mmu$ be a $s$-qubit quantum oracle, and let $\alpha$ be the probability that measuring the result of applying $\mmu$ yields an accepting state. There is a (quantum) algorithm using $O(\alpha^{-\half}\log \frac 1 \delta)$ queries to $\mmu$ and $O(\log s\log \frac 1 \delta)$ additional time that returns $s$ with $s \in S$ with probability $\ge 1 - \delta$.
\end{proposition}
\paragraph{Loading from trees.} Given a dynamic vector $x \in \R^m_{\ge 0}$ which is supported in an appropriate efficient data structure $\samptree$ (see Lemma~\ref{lem:samplertree}), and a known bound $\beta \ge \norm{x}_1$, we recall a result of \cite{GroverR02} which allows us to form a superposition of the entries in $x$ (suitably rescaled). 

\begin{lemma}
	Let $x \in \R^m_{\ge 0}$ correspond to an instance of $\samptree$, and $\beta \ge \norm{x}_1$. We can maintain a quantum oracle $\orst$ which takes $O(\log m)$ time to apply, such that the total cost of building $\orst$ after $T$ calls to $\update$ is $O(T \log m)$, and
	\[\orst \ket{0}^{\otimes (a + 1)} = \sum_{i \in [m]} \sqrt{\frac{x_i}{\beta}} \ket{0} \ket{i} + \sqrt{1 - \frac{\norm{x}_1}{\beta}} \ket{1} \ket{g}.\]
\end{lemma}
\begin{proof}
	This is implicit in \cite{GroverR02}. We first apply a $1$-qubit gate to condition on selecting from the tree (with probability $\frac{\norm{x}_1}{\beta}$), and then apply the \cite{GroverR02} procedure conditioned on the first qubit being $\ket{0}$, which controls for one qubit at a time while propagating subtree sums (provided by $\samptree$ via $\subsum$). The cost to build the circuit follows because on an $\update$ we need to change the gates corresponding to the relevant leaf-to-root path.
	
\end{proof}

\begin{corollary}\label{cor:loadax}
	Let $x \in \R^m_{\ge 0}$ correspond to an instance of $\samptree$, and let $\beta \ge \norm{x}_1$, and suppose $\ma \in \R^{m \times n}$ has $\norm{\ma}_{\max} \le 1$. We can maintain a quantum oracle $\oracle_{\ma^\top x}$ which takes $O(\log m)$ time to apply, with total building cost $O(T \log m)$ after $T$ calls to $\update$, such that for any $j \in [n]$,
	\[\oracle_{\ma^\top x} \ket{0}^{\otimes (a + 2)} \ket{j} = \ket{0} \Par{\sum_{i \in [m]} \sqrt{\frac{\ma_{ij} x_i}{\beta}} \ket{0} \ket{i} + \ket{1}\ket{g}}\ket{j}. \]
\end{corollary}
\begin{proof}
	We apply $\orma'$ (see Section~\ref{sec:prelims}) to the output of $\orst$, ignoring the additional qubit.
\end{proof}

We remark here that the additional qubit in Corollary~\ref{cor:loadax} will shortly become useful in constructing an appropriate block-encoding of a scaling of $\diag{\ma^\top x}$.

\paragraph{Polynomial approximation.} In order to give approximate Gibbs samplers for the types of dynamic vectors Algorithm~\ref{alg:main} encounters, we further require some tools from polynomial approximation theory. We first state a helper result on boundedly approximating the exponential, a variant of which was also used in \cite{vanApeldoornG19}. We provide a proof in Appendix~\ref{app:boundedexp}.

\begin{restatable}[Lemma 7,~\cite{vanApeldoornG19}]{lemma}{restatepolyhelper}\label{lem:poly-approx-helper}
	Let $\beta\ge1$, $\xi\le \frac 1 {10}$. There is a polynomial $\poly_{\beta,\xi}$ of degree $O(\beta\log \frac 1 \xi)$ such that $\max_{x \in [-1, 1]}|\poly_{\beta,\xi}(x)|\le 3$ and $\max_{x \in [-1, 0]}|\poly_{\beta,\xi}(x)-\exp(\beta x)|\le \xi$.
\end{restatable}

Next, we state a further corollary of Lemma~\ref{lem:poly-approx-helper} to be used in our rejection sampler.

\begin{corollary}\label{cor:polyshift}
	Let $B, \delta \ge 0$ and suppose $v \in \R^n$ has $\norm{v}_\infty \le B$. Further, suppose for some $c \ge 0$, $- c \le \max_{j \in [n]} v_j \le 0$.
	Let $q \in \R^n_{\ge 0}$ satisfy $q_j \in [\ell, 1]$ entrywise. Finally, define $u_j \defeq \frac{v_j}{2B}$ entrywise. There is a degree-$\Delta$ polynomial $\poly$, for $\Delta = O(B \cdot (c + \log\frac{n}{\ell\delta}))$, such that for $w_j \defeq \poly(u_j)^2 q_j$ and $z_j \defeq \exp(2B u_j)q_j$ entrywise,
	\begin{equation}\label{eq:wzclose}\norm{\frac w {\norm{w}_1} - \frac z {\norm{z}_1}}_1 \le \delta.\end{equation}
\end{corollary}
Moreover, $\max_{x \in [-1, 1]} |\poly(x)| \le \half$, and $\|w\|_1 \ge \frac{1 - \delta}{36} \|z\|_1$.
\begin{proof}
	Assume $\delta \le 2$ else the statement is clearly true. First, $u_j \in [-\half, 0]$ entrywise by the stated assumptions (since $v_j \in [-B, 0]$ entrywise). Let $\poly_{\beta, \xi}(\cdot)$ be the polynomial given by Lemma~\ref{lem:poly-approx-helper} which $\xi$-approximates $\exp(\beta \cdot)$ on $[-\frac 1 2, 0]$. We define 
	\[\poly(u) \defeq \frac 1 6 \poly_{B, \xi}\Par{u},\text{ for } \xi \defeq \frac{\delta\ell}{6n\exp(c)}.\]
	The degree bound and absolute value bound of this polynomial follows immediately from Lemma~\ref{lem:poly-approx-helper}, so it remains to show the distance bound. The guarantees of Lemma~\ref{lem:poly-approx-helper} then imply for all $j \in [n]$,
	\begin{equation}\label{eq:polyclose}\Abs{6\poly(u_j) - \exp\Par{Bu_j}} \le \xi.\end{equation}
	We further have that $u_j \le 0$, so $\exp(Bu_j) \le 1$. Hence, we also have
	\[\Abs{6\poly(u_j) + \exp\Par{Bu_j}} \le 2 + \xi \le 3. \]
	Combining yields for all $j \in [n]$,
	\begin{equation}\label{eq:polyapprox}\Abs{36\poly(u_j)^2 - \exp\Par{2Bu_j}} \le 3\xi.\end{equation}
	Next, let $y_j \defeq 36 w_j$ for all $j \in [n]$, and note that $\frac{y}{\|y\|_1} = \frac{w}{\|w\|_1}$. We bound
	\begin{equation}\label{eq:abdist}
		\begin{aligned}
			\norm{\frac w {\norm{w}_1} - \frac z {\norm{z}_1}}_1 &= \sum_{j \in [n]} \Abs{\frac{y_j}{\norm{y}_1} - \frac{z_j}{\norm{z}_1}} \le \sum_{j \in [n]} \Abs{\frac{y_j}{\norm{y}_1} - \frac{y_j}{\norm{z}_1}} + \sum_{j \in [n]} \Abs{\frac{y_j}{\norm{z}_1} - \frac{z_j}{\norm{z}_1}} \\
			&\le \Abs{1 - \frac{\norm{y}_1}{\norm{z}_1}}+ \frac{\norm{y - z}_1}{\norm{z}_1} \le \frac{2\norm{y - z}_1}{\norm{z}_1}.
		\end{aligned}
	\end{equation}
	By using the definitions of $y, z$ and \eqref{eq:polyapprox}, as well as the assumed ranges on $q$,
	\begin{align*}
		\norm{y - z}_1 \le 3n\xi,\;\norm{z}_1 \ge \ell \exp(- c).
	\end{align*}
	The second inequality used that some $v_j = 2B u_j$ is at least $-c$ by assumption. Combining the above display with \eqref{eq:abdist} and the definition of $\xi$ concludes the proof of \eqref{eq:wzclose}. Finally, using the bounds on $\norm{y - z}_1, \|z\|_1$ above shows that
	\[\|w\|_1 = \frac{1}{36}\|y\|_1 \ge \frac{1 - \delta}{36} \|z\|_1.\]
\end{proof}

\paragraph{Block-encoding.} Our approximate Gibbs oracle follows an implementation strategy pioneered by \cite{GilyenSLW19} termed ``block-encoding.'' Specifically, we follow \cite{GilyenSLW19} and say that $\mmu$, an $(a+\ell)$-qubit quantum gate, is an $\ell$-bit block-encoding of $\mm$ if the top-left $2^a \times 2^a$ submatrix of $\mmu$ is $\mm$. Block-encoded matrices admit efficient composable operations, such as the application of linear combinations and bounded polynomials. We summarize these properties in the following.

\begin{proposition}[Lemma 52, \cite{GilyenSLW19}]\label{prop:addblock}
	Let $\mmu_1$ and $\mmu_2$ be $\ell$-bit block-encodings of $\mm_1$, $\mm_2$ of the same size. There is an $O(\ell)$-bit block-encoding of $\half \mm_1 + \half \mm_2$ which takes the same asymptotic time to apply as applying $\mmu_1$ and $\mmu_2$.
\end{proposition}

\begin{proposition}[Theorem 56, \cite{GilyenSLW19}]\label{prop:polyblock}
	Let $\mmu$ be an $\ell$-bit block-encoding of $\mm$, and $\poly:[-1, 1] \to [-\half, \half]$ be a degree-$\Delta$ polynomial. There is an $O(\ell)$-bit block-encoding of $\poly(\mm)$ which can be applied in $O(\Delta)$ applications of $\mmu$ and $\mmu^\dagger$ and $O(\ell\Delta)$ additional time.
\end{proposition}

We also demonstrate that an application of Corollary~\ref{cor:loadax} yields a simple block-encoding of $\diag{\frac{\ma^\top x}{\beta}}$. A similar construction previously appeared in \cite{vanApeldoornG19}.

\begin{corollary}\label{cor:blockencodediag}
	Let $x \in \R^m_{\ge 0}$ correspond to an instance of $\samptree$, and $\beta \ge \norm{x}_1$. Let $\mm \defeq \diag{\frac{\ma^\top x}{\beta}}$ and $\mmu \defeq \oracle_{\ma^\top x}^* (\textup{SWAP}_{12} \otimes \id) \oracle_{\ma^\top x}$, where $\textup{SWAP}_{12}$ swaps the first two qubits and $\oracle_{\ma^\top x}$ is from Corollary~\ref{cor:loadax}. Then $\mmu$ is a block-encoding of $\mm$, and can be applied in time $O(\log m)$, with total building cost $O(T \log m)$ after $T$ calls to $\update$.
\end{corollary}
\begin{proof}
	Define $w_{ij} \defeq \frac{\ma_{ij} x_i}{\beta}$ for convenience.
	By the definition of $\oracle_{\ma^\top x}$, we have that
	\[\Par{\textup{SWAP}_{12} \otimes \id} \oracle_{\ma^\top x} \Par{\ket{0}^{\otimes (a + 2)} \ket{j}} = \Par{\ket{00} \sum_{i \in [m]} \sqrt{w_{ij}}\ket{i} + \ket{10}\ket{g} }\ket{j}. \]
	Hence, for $j, j' \in [n]$, we compute $\bra{j'}\bra{0}^{\otimes(a + 2)} \mmu \ket{0}^{\otimes(a + 2)}\ket{j}$ as:
	\begin{align*}
		\bra{j'}\Par{\ket{00} \sum_{i \in [m]} \sqrt{w_{ij}}\ket{i} + \ket{01}\ket{g} }^*\Par{\ket{00} \sum_{i \in [m]} \sqrt{w_{ij}}\ket{i} + \ket{10}\ket{g} }\ket{j} \\
		= \begin{cases}
			\sum_{i \in [m]} w_{ij} = \frac{[\ma^\top x]_j}{\beta} & j = j'	\\
			0 & j \neq j'
		\end{cases}.
	\end{align*}
	In particular the $\ket{01}$ and $\ket{10}$ terms disappear, and $\ket{j}$, $\ket{j'}$ are orthogonal unless $j = j'$. In the above, we required that $\sqrt{w_{ij}}^* \sqrt{w_{ij}} = w_{ij}$, which is only true if $w_{ij}$ is nonnegative. To bypass this issue, we will implement the two copies of $\oracle_{\ma^\top x}$ in slightly different ways, to obtain the correct signing. For notational clarity, we let $\oracle^L$ be the oracle which is conjugated on the left and $\oracle^R$ be the oracle on the right, such that $\mmu = (\oracle^L)^* (\textup{SWAP}_{12} \otimes \id) (\oracle^R)$. Note that $x$ is entrywise nonnegative and $\beta > 0$, and hence the only factor determining the sign of $w_{ij}$ is $\ma_{ij}$. When $\ma_{ij} \ge 0$, we will define the oracles $\orma'$ used to load $\sqrt{\ma_{ij}}$ for $\oracle^L$ and $\oracle^R$ in a consistent way (i.e.\ use the same-signed square root), so that $\sqrt{w_{ij}}^2 = w_{ij}$. When $\ma_{ij}< 0$ we will define them in an inconsistent way, so that after the conjugation operation, $-\sqrt{w_{ij}} \sqrt{w_{ij}} = w_{ij}$. We have thus shown that $\bra{0}^{\otimes(a + 2)} \mmu \ket{0}^{\otimes(a + 2)} = \mm$ which implies the first conclusion. To see the second, all our gates are reversible (arithmetic circuits are reversible, and $\orma$ is its own inverse), and hence the complexity of applying $\oracle_{\ma^\top x}^*$ is the same as $\oracle_{\ma^\top x}$.
\end{proof}

Finally, we put together the pieces and prove Proposition~\ref{prop:mainhintsample}, which we use repeatedly throughout the paper to implement our Gibbs sampling oracles.

\restatemainhintsample*
\begin{proof}
	Throughout the proof, let $\delta \gets \min(\half, \delta)$ and $B \defeq 4(\beta + \log(\frac{Cn}{\delta}))$. Also define $\ell \defeq \frac \delta n$ (following notation of Corollary~\ref{cor:polyshift}). We first observe that since $\max_{j \in [n]} [\ma^\top x]_j \le \log Z \le \max_{j \in [n]} [\ma^\top x]_j + \log n$,
	\[-\log (Cn) \le \max_{j \in [n]} [\ma^\top x]_j - \log\Par{\tZ q_j} \le 0.\]
	Here, the upper bound used that for all $j \in [n]$, $\exp([\ma^\top x]_j - \tZ q_j) = \frac{p_j}{q_j} \cdot \frac{Z}{\tZ} \le 1$ by assumption.
	Hence, for $v \defeq \ma^\top x - \log(\tZ q)$ entrywise,
	\[- c \le \max_{j \in [n]} v_j \le 0 \text{ for } c \defeq \log (Cn).\]
	Next, we note $\log(\tZ q)$ is entrywise bounded in magnitude by $\frac B 2$: 
	\begin{align*}\log(\tZ q_j) &\le \log(CZ) \le \log \Par{n \cdot \max_{j \in [n]} \exp([\ma^\top x]_j)} + \log C  \le \frac B 2,\\
		\log(\tZ q_j)&\ge \log Z + \log \frac \delta n \ge \min_{j \in [n]} [\ma^\top x]_j - \log \frac n \delta \ge -\frac B 2.
	\end{align*}
	Define $\mm_1 \defeq \diag{\frac{\ma^\top x}{2B}}$ and $\mm_2 \defeq \diag{-\frac{1}{2B} \log(\tZ q)}$. By the calculations above, we have $\normop{\mm_2} \le \half$, and similarly it is clear that $\normop{\mm_1} \le \half$ because $\norm{\ma^\top x}_\infty \le \beta$. Moreover, by using Corollary~\ref{cor:blockencodediag} with $\beta \gets B$, we obtain $\mmu_1$, a block-encoding of $\mm_1$ applicable in $O(\log m)$ time. Using a similar construction as Corollary~\ref{cor:blockencodediag}, since $q$, $B$, and $\tZ$ are all efficiently classically queriable, we obtain $\mmu_2$, a block-encoding of $\mm_2$ applicable in $O(1)$ time. Hence, Proposition~\ref{prop:addblock} yields $\mmu$, a block-encoding of 
	\[\mm_1 + \mm_2 = \textbf{diag}\Par{\frac v {2B}},\]
	which can be applied in $O(\log mn)$ time. Next, let $\poly$ be the degree-$\Delta = O(B \log\frac{Cn}{\delta})$ polynomial from Corollary~\ref{cor:polyshift}, parameterized by $B, v, c, q, \ell$ as defined earlier. Corollary~\ref{cor:polyshift} shows that $\poly: [-1, 1] \to [-\half, \half]$. Thus, Proposition~\ref{prop:polyblock} then yields $\mmu'$, a block-encoding of $\diag{\poly(\frac v {2B})}$ which can be applied in $O(\Delta \cdot \log mn)$ time. Furthermore, since $q$ and $\rho$ are efficiently classically queriable, we can define a gate $\oracle_q$ which is applicable in $O(1)$ time and acts as
	\[\oracle_q \ket{0}^{\otimes (b + 1)} = \ket{0}\sum_{j \in [n]} \sqrt{\frac{q_j}{\rho}}\ket{j} + \ket{1}\ket{g}.\]
	Applying $\mmu'$ to the output of $\oracle_q$ with appropriate ancilla qubits then yields
	\[\ket{0}^{\otimes O(1)} \sum_{j \in [n]} \sqrt{\frac{q_j \poly(u_j)^2}{\rho}} \ket{j} \ket{g_j} + \ket{g'},\text{ where } u_j \defeq \frac {v_j} {2B} \text{ for all } j \in [n].\]
	Post-selecting on the first register being the all-zeroes state and measuring on the register corresponding to $j$, we see that we obtain a sample $j \in [n]$ with probability proportional to $q_j \poly(u_j)^2$. By Corollary~\ref{cor:polyshift}, conditioned on the sample succeeding, the resulting distribution is $\delta$-close in $\ell_1$ to the distribution proportional to $q \circ \exp(v) \propto \exp(\ma^\top x)$, and hence the result is a $\delta$-approximate Gibbs oracle. Finally, we bound the query cost of the oracle. Define $w_j \defeq \poly(u_j)^2 q_j$ and $z_j \defeq \exp(v_j) q_j$ as in Corollary~\ref{cor:polyshift}. By definition of $v, \tZ$,
	\[\norm{z}_1 = \sum_{j \in [n]} \frac{\exp\Par{\Brack{\ma^\top x}_j}}{\tZ} \in \Brack{C^{-1},1}.\]
	Moreover, the last conclusion in Corollary~\ref{cor:polyshift} shows $\norm{w}_1 \ge\frac 1 {72} \norm{z}_1 \ge (72C)^{-1}$. Hence,
	\[\sum_{j \in [n]} \frac{q_j \poly(u_j)^2}{\rho} = \frac{\norm{w}_1}{\rho} \ge \frac 1 {72C\rho}. \]
	In other words, we have an oracle which we can apply in time $O(\Delta \cdot \log mn)$ which correctly returns a sample with probability $\alpha \ge \frac 1 {72C\rho}$. By applying Proposition~\ref{prop:amplif} to improve the success probability, we obtain the desired conclusion at a $O(\sqrt{C\rho}\log \frac 1 \delta)$ overhead.
	
\end{proof}

\restateotest*
\begin{proof}
	Our oracle $\otest$ is the oracle from Proposition~\ref{prop:mainhintsample}, except we will choose a sufficiently small constant value of $\delta$. It returns ``success'' when the sample is accepted by the rejection sampler after boosting by amplitude amplification. Before boosting, the success probability from Proposition~\ref{prop:mainhintsample} is $\Theta(\frac{1}{R\rho})$ where the constants in the upper and lower bounds are explicit. 
	Further, the constants from Proposition~\ref{prop:amplif} are explicit, and hence boosting by amplitude amplification improves the success probability to $\Theta(\frac{1}{\sqrt{R\rho}})$ with known constant bounds as required by the corollary statement.
\end{proof} %
\section{Bounded approximation to $\exp$ on $[-1, 1]$}\label{app:boundedexp}

Here, we give a proof of a lemma (with slightly different constants) used in the prior work \cite{vanApeldoornG19}. This section builds entirely off prior results on polynomial approximation in \cite{GilyenSLW19}; we include it for completeness because a proof was not given in \cite{vanApeldoornG19}. As a reminder, we stated and used the following result earlier when constructing our rejection sampler in Appendix~\ref{sec:quantum}.

\restatepolyhelper*

To obtain the lemma, we will utilize the following result from \cite{GilyenSLW19}.

\begin{proposition}[Corollary 66,~\cite{GilyenSLW19}]\label{prop:poly-approx-helper}

	Let $x_0\in[-1, 1]$, $r\in(0,2]$, $\delta\in(0,r]$. Let $f:[x_0-r-\delta, x_0+r+\delta]\rightarrow\mathbb{C}$ be such that $f(x_0+x) = \sum_{\ell\ge0}a_\ell x^\ell$ for all $x\in[-r-\delta, r+\delta]$. Suppose $B>0$ is such that $\sum_{\ell\ge0}(r+\delta)^\ell|a_\ell|\le B$ and let $\epsilon\in(0,\frac 1 {2B}]$. There is a polynomial $P$ (see~\Cref{app:stablepoly} for its numerically stable implementation) of degree $O\left(\frac{1}{\delta}\log\frac{B}{\epsilon}\right)$ such that 
	\begin{align*}
	\max_{x\in [x_0-r, x_0+r]}|f(x)-P(x)|\le \epsilon
	\text{ and }
	\max_{x\in[-1,1]}|P(x)|\le \epsilon+B.	
	\end{align*}
\end{proposition}

\begin{proof}[Proof of~\Cref{lem:poly-approx-helper}]
We apply Proposition~\ref{prop:poly-approx-helper} with $f(x) \defeq \exp(\beta x)$ which has a convergent Taylor series everywhere, and the parameter settings $x_0 = -1$, $r = 1$, $\delta = \frac 1 \beta$, $B = e$. We have that $f(x_0+x) = \sum_{\ell\ge0}\exp(-\beta)\frac{\beta^\ell\cdot x^\ell}{\ell!}=\sum_{\ell\ge0}a_\ell x^\ell$ with $a_\ell = \exp(-\beta)\frac{\beta^\ell}{\ell!}$ for any integer $\ell\ge0$. We also check that our choice of $B$ is valid, via \[\sum_{\ell\ge0}(r+\delta)^\ell|a_\ell| = \exp(-\beta)\sum_{\ell\ge0}\Par{1+\frac{1}{\beta}}^\ell\frac{\beta^\ell}{\ell!} = \exp(-\beta)\sum_{\ell\ge0}\frac{(\beta+1)^\ell}{\ell!} = \exp(\beta+1-\beta) = e.\]

Hence by~\Cref{prop:poly-approx-helper}, we have for any $\xi \le \frac 1 {2e}$, there is a polynomial $P$ of degree $O(\beta\log \frac 1 \xi)$ such that $\max_{x\in[-2,0]}|\exp(\beta x)-P(x)|\le \epsilon$ and $\max_{x\in[-1,1]}|\tilde{P}(x)|\le e+\frac 1 6 + \xi \le 3$. 
\end{proof}

\section{Numerically stable implementation of polynomial approximation}\label{app:stablepoly}

Throughout this section, let $\Delta = O(\frac 1 \eps \log^2(\frac {mn} \eps))$ be the degree of the polynomial used in the proof of Proposition~\ref{prop:mainhintsample} in Appendix~\ref{sec:quantum} (specifically, constructed in the proof of Proposition~\ref{prop:mainhintsample}, where we have $C = O(1)$ and $\delta = O(\eps)$ in our applications). The polynomial we use is constructed via a decomposition in the Fourier basis (see Lemmas 57 and 65, \cite{GilyenSLW19}). It is not immediate that this polynomial transform can be implemented stably in finite-precision arithmetic, within the quantum singular value transformation framework of \cite{GilyenSLW19}, which is used in the proof of Proposition~\ref{prop:mainhintsample}. However, \cite{Haah19} shows that given such a decomposition in the Fourier basis, we can obtain a numerically-stable implementation of the polynomial transformation required as a quantum circuit up to additive error $\xi$, in time
\[O\Par{\Delta^3 \log\Par{\frac{\Delta}{\xi}}}.\]
In our setting (in the proof of Proposition~\ref{prop:mainhintsample}), it is straightforward to check that $\xi = \textup{poly}(m, n, \eps^{-1})$. This construction results in the additive term in Theorem~\ref{thm:quantumgames}. 
\end{appendix}

\end{document}